\PassOptionsToPackage{hyphens}{url}
\documentclass[a4paper,USenglish,cleveref,autoref,thm-restate,nolineno]{socg-lipics-v2021}

\pdfoutput=1 
\hideLIPIcs  

\graphicspath{{./Figures/}}

\usepackage[shortlabels,inline]{enumitem}
\usepackage{bm}
\usepackage{mathtools}
\usepackage{overpic}
\usepackage{subcaption} 

\newcommand{\myref}[1]{{\color{lipicsGray}\sffamily\bfseries\ref{#1}}}

\newcommand{\defeq}{\vcentcolon=} 
\newcommand{\eqdef}{=\vcentcolon}

\newcommand{\complex}{\mathscr{C}}

\newcommand{\alttensor}{\mathsf{U}}
\newcommand{\antipode}{\mathsf{S}}
\newcommand{\alttree}{T}
\newcommand{\bags}{\mathscr{B}}

\newcommand{\comult}{\mathsf{\Delta}}
\newcommand{\diag}{\mathscr{D}}

\newcommand{\dual}{\Gamma}                         

\newcommand{\face}{F}
\newcommand{\field}{\mathbb{F}}

\newcommand{\group}{\mathcal{G}}
\newcommand{\hbody}{\mathscr{H}} 

\newcommand{\hopf}{\mathcal{H}}
\newcommand{\hopfvec}{\bm{H}}

\newcommand{\linmap}{\mathsf{A}}

\newcommand{\manifold}{\mathscr{M}}
\newcommand{\mult}{\mathsf{M}}

\newcommand{\network}{\mathscr{K}}
\newcommand{\altnetwork}{\mathscr{N}}

\newcommand{\alphacurves}{\bm{\alpha}}
\newcommand{\betacurves}{\bm{\beta}}
\newcommand{\gammacurves}{\bm{\gamma}}
\newcommand{\oalpha}{\vec{\bm{\alpha}}}
\newcommand{\obeta}{\vec{\bm{\beta}}}
\newcommand{\odiag}{\vec{\diag}}
\newcommand{\otri}{\vec{\tri}}

\newcommand{\osurface}{\vec{\surface}}

\newcommand{\poly}{\operatorname{poly}}

\newcommand{\vectorspace}{\bm{V}}
\newcommand{\dualspace}{\bm{V}^\ast}

\newcommand{\spine}{S}
\newcommand{\surface}{\mathscr{S}}

\newcommand{\tensor}{\mathsf{T}}

\newcommand{\trace}{\operatorname{Tr}}

\newcommand{\vect}{\mathsf{v}}

\newcommand{\treedecomp}{\mathscr{X}}
\newcommand{\tri}{\mathscr{T}}

\newcommand{\R}{\mathbb{R}}

\newcommand{\chd}{\operatorname{chd}}          
   
\newcommand{\vc}{\operatorname{vc}}             
\newcommand{\tw}{\operatorname{tw}}            

\newcommand{\M}{\manifold}

\newcommand{\val}{\operatorname{val}}

\newcommand{\maxvalence}{\bm{\Delta}}
\newcommand{\maxval}{\Delta}
\newcommand{\vertices}{\mathbf{v}}
\newcommand{\treewidth}{\mathbf{tw}}

\title{On Sparse Representations of 3‑Manifolds} 

\titlerunning{On Sparse Representations of 3‑Manifolds} 

\author{Kristóf Huszár}{Institute of Geometry, Graz University of Technology, Austria \and \url{https://kristofhuszar.github.io} }{kristof.huszar@tugraz.at}{https://orcid.org/0000-0002-5445-5057}{}

\author{Clément Maria}{Inria d'Université Côte d'Azur, France \and \url{https://www-sop.inria.fr/members/Clement.Maria/} }{ clement.maria@inria.fr}{https://orcid.org/0000-0002-2007-2584}{Partially supported by the ANR project ANR-20-CE48-0007 (AlgoKnot).}

\authorrunning{K. Huszár and C. Maria} 

\Copyright{Kristóf Huszár and Clément Maria} 



\category{} 

\relatedversion{} 

\nolinenumbers 

\EventEditors{John Q. Open and Joan R. Access}
\EventNoEds{2}
\EventLongTitle{42nd Conference on Very Important Topics (CVIT 2016)}
\EventShortTitle{CVIT 2016}
\EventAcronym{CVIT}
\EventYear{2016}
\EventDate{December 24--27, 2016}
\EventLocation{Little Whinging, United Kingdom}
\EventLogo{}
\SeriesVolume{42}
\ArticleNo{23}

\begin{document}

\maketitle

\begin{abstract}
3-manifolds are commonly represented as triangulations, consisting of abstract tetrahedra whose triangular faces are identified in pairs. The combinatorial sparsity of a triangulation, as measured by the treewidth of its dual graph, plays a fundamental role in the design of parameterized algorithms. In this work, we investigate algorithmic procedures that transform or modify a given triangulation while controlling specific sparsity parameters. First, we revisit a standard, linear-time algorithm that converts a given triangulation into a Heegaard diagram of the underlying 3-manifold, showing that the construction preserves treewidth. We apply this construction to exhibit a fixed-parameter tractable framework for computing Kuperberg's quantum invariants of 3-manifolds. Second, we present a quasi-linear-time algorithm that retriangulates a given triangulation into one with maximum edge valence of at most nine, while only moderately increasing the treewidth of the dual graph. Combining these two algorithms yields a quasi-linear-time algorithm that produces, from a given triangulation, a Heegaard diagram in which every attaching curve intersects at most nine others.

\end{abstract}

\section{Introduction}
\label{sec:intro}


Structural properties of triangulations---which are commonly used to encode 3-manifolds both in theory and in practice---can dramatically affect the feasibility of computations. Over the past decade, several fixed-parameter tractable (FPT) algorithms have been developed that efficiently solve provably hard problems on 3-manifolds, provided they are represented by triangulations that are sufficiently ``thin'' \cite{burton2017courcelle, burton2016parameterized, burton2018algorithms, pettersson2014fixed, burton2013complexity}.\footnote{See \cite{burton2018homfly, makowsky2005coloured, makowsky2003parameterized, maria2021parametrized} for related FPT algorithms in knot theory, and the survey \cite{lackenby2020algorithms} for a broader context.} On input triangulations with bounded \emph{treewidth}\footnote{Treewidth is a graph parameter that quantifies the similarity of a given graph to any tree. We define the treewidth of a triangulation $\tri$ as the treewidth of its dual graph $\dual(\tri)$. See \Cref{sec:prelims} for details.} these algorithms run in time polynomial in the number of tetrahedra.\footnote{Many of these algorithms have also been implemented in the topology software \texttt{Regina} \cite{regina}.}

Motivated by these algorithms, several recent papers have investigated the quantitative relationship between the treewidth (and other width parameters) of triangulations and the corresponding quantities in various related representations of 3-manifolds. It has been shown that 3-manifolds with small Heegaard genus \cite{he2025algorithm,huszar2019manifold}, as well as hyperbolic 3-manifolds with small volume \cite{huszar2022pathwidth,maria2019treewidth}, admit triangulations with small treewidth. At the same time, for certain 3-manifolds, a large Heegaard genus or a ``complicated'' JSJ decomposition can entirely preclude the existence of triangulations with small treewidth \cite{huszar2023width-DCG,huszar2019treewidth}.\footnote{See \cite{mesmay2019treewidth, lunel2023structural} for related structural results about knots and links and their diagrams.}

In this work, we also focus on algorithmic transformations of 3-manifold triangulations with a view toward parameterized algorithms and sparsity. We provide two main contributions. First, we revisit and analyze a classical construction that turns a triangulation $\tri$ of a closed 3-manifold $\manifold$ into a Heegaard splitting\footnote{A Heegaard splitting is a decomposition of a 3-manifold into two identical handlebodies, cf.\ \Cref{ssec:mfds}.} and show that it yields a Heegaard diagram of $\manifold$, whose size and treewidth are linearly bounded by those of $\tri$. More concretely, we prove

\begin{theorem}
\label{thm:heegaard-diag-tw}
Let $\tri$ be a triangulation of a closed, orientable $3$-manifold $\manifold$ with $n$ tetrahedra and dual graph $\dual(\tri)$. Let $\diag \defeq \diag(\tri)$ be the Heegaard diagram of $\manifold$ induced by $\tri$. Then, for the number of vertices and for the treewidth\footnote{The number $V(\diag)$ of vertices and the treewidth $\tw(\diag)$ of a Heegaard diagram $\diag = (\surface,\alphacurves,\betacurves)$ are, by definition, respectively equal to those of its underlying graph, which is the 4-regular multigraph geometrically obtained by taking the union of the $\alpha$- and $\beta$-curves and placing a vertex at each crossing.} of $\diag$ and $\dual(\tri)$ we have

\noindent\begin{subequations}
\hspace*{\fill}%
\parbox{0.25\textwidth}{%
\begin{align}
	|V(\diag)| = 6n
	\label{eq:diagram-vertices}
\end{align}%
}\hfill%
\parbox{0.1\textwidth}{%
\begin{align*}	
	~and
\end{align*}%
}\hfill%
\parbox{0.42\textwidth}{%
\begin{align}
	\tw(\diag) \leq 12\tw(\dual(\tri)) + 11.
	\label{eq:diagram-width}
\end{align}%
}%
\hspace*{\fill}
\end{subequations}

\noindent Moreover, given $\tri$, the induced Heegaard diagram $\diag$ can be constructed in $O(n)$ time.
\end{theorem}

Guided by this result, in \Cref{sec:tensor} we investigate the complexity of computing Kuperberg's quantum invariants of oriented 3-manifolds~\cite{kuperberg1991involutory}, which are obtained from a Heegaard diagram by evaluating an associated tensor network. First, we show that this construction is width-preserving (\Cref{lem:kuperberg-width}). This insight, combined with \Cref{thm:heegaard-diag-tw} and results on the complexity of evaluating tensor networks~\cite{markov2008simulating,ogorman2019parameterization} (cf.\ \Cref{thm:contraction}), provides a framework for computing Kuperberg's invariants from triangulations that is FPT in the treewidth (\Cref{thm:kuperberg-fpt}).

\begin{remark*}
The converse of \Cref{thm:heegaard-diag-tw}, i.e., building a triangulation from a given Heegaard splitting in a width-preserving way, has been studied in \cite{huszar2019manifold} and, very recently, in \cite{he2025algorithm}, where an explicit algorithm for constructing a triangulation from a Heegaard diagram is provided.
\end{remark*}

Our second contribution investigates the interplay between the \emph{maximum edge valence} and the treewidth of a 3-manifold triangulation $\tri$. Analogous to the notion of vertex degree in graphs, the \emph{valence of an edge} in $\tri$ is defined as the number of tetrahedra in $\tri$ that contain it. While treewidth can be viewed as a measure of \emph{global} sparsity, the maximum edge valence captures \emph{local} sparsity. The maximum edge valence of a triangulation can also provide valuable insights into geometric properties of the underlying 3-manifold \cite[Sec.\ 3.6]{frick2015thesis}. 
Having a triangulation $\tri$ with small maximum edge valence can be computationally advantageous even when $\tri$ already has small treewidth; see Appendix~\ref{app:edge_valence_expe} for a discussion.

Every closed orientable 3-manifold has a triangulation with maximum edge valence at most six~\cite{brady2004bounding}. Frick has shown by an explicit construction that any triangulation $\tri$ can be modified into a triangulation $\tri^\dagger$ of the same manifold with edge valences at most nine~\cite[Theorem 3.36]{frick2015thesis}. However, this construction may significantly increase the treewidth of the resulting triangulation, because if $\tri$ contains edges with large valence, the retriangulation procedure introduces large grid-like structures in $\tri^\dagger$. In \Cref{sec:retri}, we show how to circumvent this and obtain a different retriangulation $\tri^\ast$ that also achieves $\maxval(\tri^\ast) \leq 9$, while controlling the increase in treewidth by a polylogarithmic factor of the maximum edge valence of $\tri$.

\begin{theorem}
\label{thm:generalalgo-intro}
There is an algorithm which, given a triangulation $\tri$ of a closed $3$-manifold $\M$ with $n$ tetrahedra, $\vertices$ vertices, $\tw(\dual(\tri)) = \treewidth$, and maximum edge valence $\maxval(\tri) = \maxvalence$, constructs a triangulation $\tri^\ast$ of $\M$ with $O((n+\vertices)\cdot\poly(\log_2(\maxvalence)))$ vertices and tetrahedra, $\tw(\dual(\tri^\ast)) \leq \treewidth \cdot \poly(\log_2(\maxvalence))$, and $\maxval(\tri^\ast) \leq 9$. The algorithm runs in $O(n \log\log n)$ time.
\end{theorem}

Finally, the execution of the algorithms in \Cref{thm:generalalgo-intro} and \Cref{thm:heegaard-diag-tw} (in this order) yields the following corollary, which we believe to be of independent interest.

\begin{corollary}
There is a quasi-linear time algorithm that, given a triangulation $\tri$ of a closed, orientable $3$-manifold $\M$ with $n$ tetrahedra, $\tw(\dual(\tri)) = \treewidth$, and maximum edge valence $\maxval(\tri) = \maxvalence$, constructs a Heegaard diagram $\diag = (\surface,\alphacurves,\betacurves)$ of $\M$ of size and genus $O(n \poly(\log(\maxvalence)))$, treewidth $O(\treewidth \cdot \poly(\log(\maxvalence))$, where any $\alpha$-curve intersects at most $9$ $\beta$-curves, and any $\beta$-curve intersects at most $3$ $\alpha$-curves.
\end{corollary}

\section{Preliminaries}
\label{sec:prelims}


\subsection{Graphs and their treewidth}
\label{ssec:graphs}

By a \emph{graph} we generally mean a \emph{multigraph} $G = (V,E)$ with a finite set $V = V(G)$ of \emph{vertices} and of a multiset $E = E(G)$ of two-element submultisets of $V$ called \emph{edges}. A \emph{loop} is an edge of the form $\{v,v\}$, and a \emph{multiedge} is an edge with multiplicity larger than one. A graph without loops or multiedges is \emph{simple}. The \emph{degree} $\deg(v)$ of $v\in V$ is the number of edges containing $v$, counted with multiplicity. A graph is \emph{$k$-regular} if $\deg(v) = k$ for all $v \in V$.

Introduced in \cite{robertson1986graph} (cf.\ \cite{eppstein2025treewidth}), the \emph{treewidth} of a graph measures its similarity to any tree. A \emph{tree decomposition} $\treedecomp$ of a graph $G=(V,E)$ is a pair $\treedecomp=(\alttree,\bags)$ of a tree $\alttree=(I,F)$ and a collection $\bags=\{B_i:i \in I\}$ of {\em bag}s, where $B_i \subseteq V$ and the following three properties hold:
 \begin{enumerate*}
	\item\label{itm:deftreedec1} $\bigcup_{i \in I} B_i = V$ (vertex coverage),
	\item\label{itm:deftreedec2} for every $\{u,v\} \in E$ there exists $i \in I$ with $\{u,v\} \subseteq B_i$ (edge coverage), and
	\item\label{itm:deftreedec3} for every $v \in V$, the set $T_v = \{i \in I:v \in B_i\}$ spans a connected subtree of $\alttree$ (subtree property).
\end{enumerate*}
The \textit{width} of the tree decomposition $\treedecomp$ equals $\max_{i \in I}|B_i|-1$, and the \emph{treewidth} $\tw(G)$ is defined as the smallest width of any tree decomposition of $G$.

\subsection{3-Manifolds and their representations}
\label{ssec:mfds}

A \emph{$d$-dimensional topological manifold with boundary} (hereafter \emph{$d$-manifold}) $\manifold$ is a topological space\footnote{Some technical requirements about the space are intentionally omitted here, cf.\ \cite[Definition 1.1.1]{schultens2014introduction}.} locally homeomorphic to $\mathbb{R}^d$, i.e.,\ every point of $\manifold$ has a neighborhood homeomorphic to $\mathbb{R}^d$ or to the upper half-space $\{(x_1,\ldots,x_d) \in \mathbb{R}^d : x_d \geq 0\}$. The \emph{boundary} $\partial\manifold$ of $\manifold$ consists of all points of $\manifold$ that have no open neighborhood homeomorphic to $\mathbb{R}^d$. In this paper we only consider compact, orientable $d$-manifolds with $d \leq 3$. A manifold $\manifold$ is \emph{closed} if it is compact and $\partial\manifold = \emptyset$. A \emph{surface} is just a $2$-manifold.  We consider manifolds up to homeomorphism. See \cite{schultens2014introduction} for general background on 3-manifolds.

\subparagraph*{Triangulations and their dual graphs}
As is common in algorithmic topology, we define a ({$3$-dimensional) \emph{triangulation} as a quotient space $\tri = \Sigma / \Phi$, where $\Sigma = \{\sigma_1,\ldots,\sigma_n\}$ is a set of labeled and vertex-labeled abstract tetrahedra, and $\Phi = \{\varphi_1,\ldots,\varphi_m\}$ is a set of simplicial isomorphisms, called \emph{gluing maps}, each of which identifies two triangles of the simplices in $\Sigma$.\footnote{Each triangle is identified with at most one other triangle. We allow two pairs of triangles to be glued between the same pair of tetrahedra, as well as gluing two triangles within a single tetrahedron. Hence the triangulations considered here are not simplicial complexes but \emph{regular simplicial cell complexes}. Nevertheless, the second barycentric subdivision $\tri''$ of any triangulation $\tri$ is always a simplicial complex.} For $i \in \{0,1,2,3\}$, let $\tri(i)$ denote the set of $i$-dimensional faces of the resulting quotient $\tri$, and let $\tri^{(i)}$, its \emph{$i$-skeleton}, be $\tri^{(i)} \defeq \cup_{j=0}^{i} \tri(j)$. The \emph{valence} $\val(e)$ of an edge $e \in \tri(1)$ is the number of tetrahedra in $\tri$ that contain it and $\maxval(\tri) \defeq \max_{e\in\tri(1)}\val(e)$. $\tri$ \emph{triangulates} the manifold $\manifold$ if the geometric realization of $\tri$ is homeomorphic to $\manifold$. Triangulations of surfaces are defined analogously. Every surface and 3-manifold has a triangulation \cite{moise1952affine,rado1925riemannschen}.

\newpage

Given a triangulation $\tri = \Sigma / \Phi$, its \emph{dual graph} (or \emph{face-pairing graph}) $\dual(\tri)$ is the graph with $V(\dual(\tri)) = \Sigma$, and $\{\sigma_1,\sigma_2\} \in E(\dual(\tri))$ if and only if there is a gluing map $\varphi \in \Phi$ that identifies a triangle of $\sigma_1$ with a triangle of $\sigma_2$ (where $\sigma_1$ and $\sigma_2$ may coincide). See \Cref{fig:tetrahedra}. We define the \emph{treewidth of a triangulation} $\tri$ as the treewidth of its dual graph $\dual(\tri)$.

\begin{figure}[ht]
	\vspace{.5em}
	\centering
	\begin{overpic}[scale=.85]{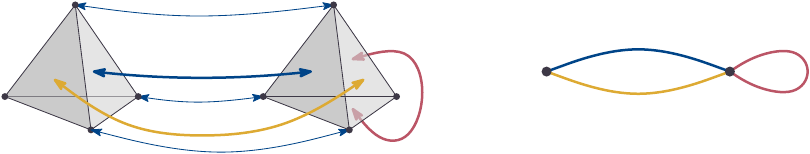}
		\put (-0.7,4) {\small{$0$}}
		\put (16.7,4) {\small{$2$}}
		\put (10.55,-0.2) {\small{$1$}}
		\put (7.25,19) {\small{$3$}}
		\put (31.7,4) {\small{$0$}}
		\put (48.9,4) {\small{$2$}}
		\put (42.55,-0.2) {\small{$1$}}
		\put (41.9,19) {\small{$3$}}
		\put (1,14.5) {\small{$\sigma_1$}}
		\put (46,14.5) {\small{$\sigma_2$}}
		\put (24,10.75) {\small{$\varphi_1$}}
		\put (22.75,3.6) {\small{$\varphi_2$}}
		\put (53,8) {\small{$\varphi_3$}}
		\put (66.5,6.75) {\small{$\sigma_1$}}
		\put (89,6.75) {\small{$\sigma_2$}}
		\put (100.75,9.5) {\small{$\varphi_3$}}
		\put (78,14.25) {\small{$\varphi_1$}}
		\put (78,4.75) {\small{$\varphi_2$}}
		\put (24,-4) {\small{$\tri$}}
		\put (80,-4) {\small{$\dual(\tri)$}}
	\end{overpic}

	\vspace{1.5em}

	\caption{A triangulation $\tri = \Sigma / \Phi$ with tetrahedra $\Sigma =\{\sigma_1,\sigma_2\}$, gluing maps $\Phi = \{\varphi_1,\varphi_2,\varphi_3\}$ and dual graph $\dual(\tri)$. $\varphi_1$ is explicitly defined as $\sigma_1(123) \leftrightarrow \sigma_2(103)$. Based on \cite[Figure 1]{huszar2023width-DCG}.}
	\label{fig:tetrahedra}
\end{figure}

\subparagraph*{Handlebodies and Heegaard splittings}
A \emph{handlebody} $\hbody$ is a compact $3$-manifold homeomorphic to a regular neighborhood of a bouquet of circles embedded in $\mathbb{R}^3$. 
Every closed, orientable $3$-manifold $\manifold$ admits a \emph{Heegaard splitting} \cite{heegaard1898thesis,heegaard1916analysis}, i.e., a decomposition into two homeomorphic handlebodies $\hbody_1$ and $\hbody_2$ with a shared boundary. That is, $\manifold = \hbody_1 \cup \hbody_2$ and $\hbody_1 \cap \hbody_2 = \partial\hbody_1 = \partial\hbody_2 = \surface$. The surface $\surface$ is called a \emph{splitting surface}, and its genus $g(\surface)$ is referred to as the \emph{genus of the Heegaard splitting}. The minimum genus of any Heegaard splitting of $\manifold$ is called the \emph{Heegaard genus} of $\manifold$. We refer to \cite{scharlemann2002heegaard} for an extensive survey.

\subparagraph*{Heegaard diagrams} Heegaard diagrams were developed alongside Heegaard splittings to describe them. They also provide the basis for computing Kuperberg's invariants (\Cref{sec:tensor}). We follow the definitions of \cite[Section~4]{kuperberg1991involutory}. In particular, we consider Heegaard diagrams in a sense more flexible than usual. A \emph{handlebody diagram} is an ordered pair $(\surface,\gammacurves)$ of a closed orientable surface $\surface$ and a collection $\gammacurves = \{\gamma_1,\ldots,\gamma_\ell\}$ of pairwise disjoint simple closed curves embedded in $\surface$ that divide $\surface$ into planar regions.\footnote{Because of this requirement the number $\ell$ of curves in $\gammacurves$ is at least the genus $g(\surface)$ of the surface $\surface$, and $\ell = g(\surface)$ is commonly assumed in the literature. However, we will generally have $\ell > g(\surface)$.} (\Cref{fig:hbody-diagrams}). 
\begin{figure}[ht]
	\centering
	\includegraphics[height=2.5cm]{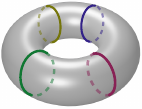}\hfill\includegraphics[height=2.5cm]{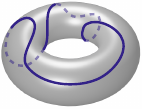}\hfill \includegraphics[height=2.5cm]{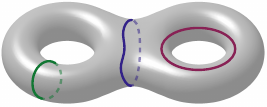}

	\caption{Illustrations of some handlebody diagrams. The surfaces are meant to be hollow.}
	\label{fig:hbody-diagrams}
\end{figure}
A handlebody diagram $(\surface,\gammacurves)$ can be seen as a prescription for gluing a handlebody to one side of $\surface$. Specifically, consider the thickening $\surface \times [0,1]$, and for each $i \in \{1,\ldots,\ell\}$ attach to its \emph{upper boundary} $\surface \times \{1\}$ a thickened disk (called a \emph{$2$-handle}) along the curve $\gamma_i$.\footnote{More precisely, we attach the 2-handle along the copy of $\gamma_i$ in $\surface \times \{1\}$.} Because of the assumption about $\gammacurves$, each resulting new boundary component is a 2-sphere, each of which we fill in with a 3-ball. (Importantly, these fillings are unique up to isotopy; see \cite[Lemma 2.5.3]{schultens2014introduction}.) The union of the attached 2-handles and the filling 3-balls is a handlebody $\hbody$, which we have just glued to $\surface$.

With that, a \emph{Heegaard diagram} is a triple $\diag = (\surface,\alphacurves,\betacurves)$ such that $(\surface,\alphacurves)$ and $(\surface,\betacurves)$ are handlebody diagrams relative to different sides of $\surface$. That is, the $2$-handles corresponding to $\alphacurves$ (resp.\ $\betacurves$) are attached to the upper boundary $\surface \times \{1\}$ (resp.\ the \emph{lower boundary} $\surface \times \{0\}$) of the thickening $\surface \times [0,1]$. We refer to the curves in $\alphacurves$ (resp.\ $\betacurves$) as \emph{$\alpha$-curves} (resp.\ \emph{$\beta$-curves}) and assign the same color to each family in illustrations; see \Cref{fig:heegaard-diagrams}. A Heegaard diagram $\diag$ can naturally be viewed as a 4-regular, surface-embedded graph, obtained geometrically by taking the union of the $\alpha$- and $\beta$-curves and placing a vertex at each crossing.

\begin{figure}[ht]
	\centering
	\hspace{1cm}\includegraphics[height=2.5cm]{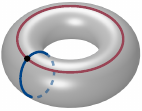}\hfill\includegraphics[height=2.5cm]{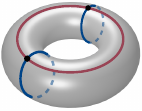}\hfill \includegraphics[height=2.5cm]{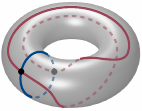}\hspace{1cm}

	\caption{Two Heegaard diagrams of the $3$-sphere, and one of the real projective $3$-space $\mathbb{RP}^3$.}
	\label{fig:heegaard-diagrams}
\end{figure}

\section{From triangulations to Heegaard diagrams}
\label{sec:algo}


In this section, we revisit a classical construction that converts a triangulation $\tri$ into a corresponding Heegaard diagram $\diag = (\surface,\alphacurves,\betacurves)$. We prove that this construction preserves both size and width up to a constant factor, as stated in \Cref{thm:heegaard-diag-tw}. Moreover, using standard data structures, the conversion $\tri \leadsto \diag$ can be performed in time linear in the size of $\tri$. We first review the construction via handles, then prove the inequalities (\ref{eq:diagram-vertices}) and (\ref{eq:diagram-width}) in \Cref{thm:heegaard-diag-tw}, and finally discuss the algorithmic aspects of the construction.

\subparagraph{Handle decompositions} Any triangulation $\tri$ of a 3-manifold $\manifold$ induces a \emph{canonical handle decomposition} $\chd(\tri)$ of $\manifold$, in which the \emph{$k$-handles} (all 3-balls) correspond to the $(3-k)$-simplices of $\tri$. Informally, $\chd(\tri)$ is the ``thickening'' of $\tri$ (see \Cref{fig:chd}). For each tetrahedron $\sigma \in \tri(3)$, take a 0-handle centered at its barycenter. For each triangle $t \in \tri(2)$, attach a 1-handle along a pair of disjoint disks connecting the 0-handles of the tetrahedra incident to $t$. For each edge $e \in \tri(1)$, attach a 2-handle along an annulus to the union of the 0- and 1-handles corresponding to the tetrahedra and triangles incident to $e$. Finally, for each vertex $v \in \tri(0)$, attach a 3-handle along a 2-sphere to fill the void surrounding $v$.

\begin{figure}[ht]
	\centering
	\begin{minipage}[t]{.66\textwidth}
		\centering
		\begin{minipage}{.31\textwidth}
		\centering
		\includegraphics[scale=.75]{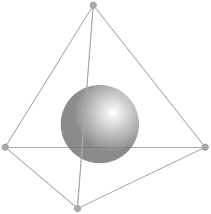}
		
		\smallskip
		
		{\footnotesize$0$-handle\vphantom{$\diag(\tri)|_{\sigma}$}}
		\end{minipage}\hfill%
		\begin{minipage}{.31\textwidth}
		\centering
		\includegraphics[scale=.75]{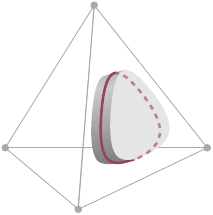}
		
		\smallskip
		
		{\footnotesize$1$-handle\vphantom{$\diag(\tri)|_{\sigma}$}}
		\end{minipage}\hfill%
		\begin{minipage}{.36\textwidth}
		\centering
		\includegraphics[scale=.75]{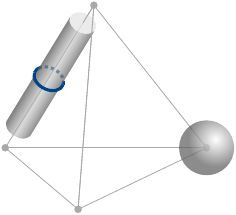}
		
		\smallskip
		
		{\footnotesize$2$- and $3$-handles\vphantom{$\diag(\tri)|_{\sigma}$}}
		\end{minipage}%

		\medskip

		\subcaption{The building blocks of the canonical handle decomposition $\chd(\tri)$. Meridian curves of the 1- and the 2-handle are also drawn.}
		\label{fig:chd}
	\end{minipage}\hfill%
	\begin{minipage}[t]{.32\textwidth}
		\centering
		\begin{minipage}{.99\textwidth}
		\centering
		\includegraphics[scale=.75]{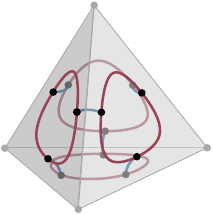}
		
		\smallskip
		
		{\footnotesize $\diag(\tri)|_{\sigma}$}
		\end{minipage}%
		
		\medskip

		\subcaption{\begin{nolinenumbers}The $\alpha$- and $\beta$-curves of $\diag(\tri)$ restricted to a tetrahedron $\sigma$ of $\tri$.\end{nolinenumbers}}
		\label{fig:chd+heegaard-local}
	\end{minipage}
	\caption{}
	\label{fig:chd+heegaard}
\end{figure}

\subparagraph{The Heegaard diagram induced by $\bm{\protect\tri}$} The union $\hbody_1$ of the 0- and 1-handles and the union $\hbody_2$ of the 2- and 3-handles of $\chd(\tri)$ are both handlebodies that together form a Heegaard splitting of the 3-manifold triangulated by $\tri$. Letting $\surface \defeq \hbody_1 \cap \hbody_2$ and taking the meridional curves of the 2-handles and those of the 1-handles as $\alpha$- and $\beta$-curves, respectively (\Cref{fig:chd}), we obtain the Heegaard diagram $\diag(\tri) = (\surface, \alphacurves, \betacurves)$, where $\alphacurves = \{\alpha_e : e \in \tri(1)\}$ and $\betacurves = \{\beta_t : t \in \tri(2)\}$. We call $\diag(\tri)$ the \emph{Heegaard diagram induced by the triangulation} $\tri$.

\begin{proof}[Proof of inequalities (\ref{eq:diagram-vertices}) and (\ref{eq:diagram-width}) in \Cref{thm:heegaard-diag-tw}]

For \eqref{eq:diagram-vertices}, note that each triangle of $\tri$ contains exactly three vertices of the Heegaard diagram $\diag \defeq \diag(\tri)$ induced by $\tri$ (\Cref{fig:chd+heegaard-local}), and each tetrahedron has four triangular faces identified in pairs, hence $|V(\diag)| = 3 \cdot 4n / 2 = 6n$. To show \eqref{eq:diagram-width}, take an optimal tree decomposition $\treedecomp=(T,\bags)$ of $\dual(\tri)$, where $T=(I,F)$ is a tree, $\bags = \{B_i : i \in I \}$ is the corresponding set of bags, and $\max\{|B_i| : i \in I \}=\tw(\dual(\tri))+1$. We turn $\treedecomp$ into a tree decomposition $\treedecomp'=(T', \bags')$ of  $\diag$ as follows.
\begin{itemize}
	\item $T' = T$, i.e., the tree structure of the decomposition remains the same.
	\item $\bags' = \{B'_i : i \in I\}$ and for each $i \in I$, the bag $B'_i$ is obtained from $B_i$ by replacing each $\sigma \in B_i$ with the (at most) 12 vertices of $\diag$ incident to $\sigma$, see \Cref{fig:chd+heegaard-local}.
\end{itemize}

Clearly, $|B'_i| \leq 12\,|B_i|$. We now check that $\treedecomp'=(T',\bags')$ is a tree decomposition of $\diag$.

\begin{description}
	\item[\emph{Vertex coverage}] Consider a vertex $v \in V(\diag)$. By construction, there exists a tetrahedron $\sigma \in \tri$, such that $v$ is incident to $\sigma$. Thus, for any $i \in I$, $v \in B'_i \Leftrightarrow \sigma \in B_i$. However, since $\bigcup_{i \in I}B_i = V(\dual(\tri))$, this implies $\bigcup_{i \in I}B'_i = V(\diag)$ as well.
	\item[\emph{Edge coverage}]  Let $e=\{u,v\} \in E(\diag)$ be any arc of the diagram. Again, by construction of $\diag$, there exists a tetrahedron $\sigma \in \tri$ that entirely contains $e$. Let $B_i$ be a bag that contains $\sigma$. Then $B'_i$ must contain both $u$ and $v$.
	\item[\emph{Subtree property}] Fix a vertex $v \in V(\diag)$ and consider $I'_v \defeq \{i \in I: v \in B'_i\}$. We need to show that $I'_v$ spans a connected subtree of $T' = T$. Let $\sigma, \tau \in \tri$ be the tetrahedra containing $v$ on their common triangular face (note that $\sigma$ and $\tau$ may coincide) and define $I_\sigma \defeq \{i \in I: \sigma \in B_i\}$ and $I_\tau \defeq \{i \in I: \tau \in B_i\}$. From the definition of $\bags'$ it follows that $I'_v = I_\sigma \cup I_\tau$.
Since $\treedecomp$ is a tree decomposition of $\dual(\tri)$, each of $I_\sigma$ and $I_\tau$ spans a connected subtree in $T'$. Moreover, $I_\sigma \cap I_\tau \neq \emptyset$: either $\sigma = \tau$, or $\{\sigma, \tau\} \in E(\dual(\tri))$, so there exists a bag $B_j$ containing both $\sigma$ and $\tau$. Thus, the connected subtrees spanned by $I_\sigma$ and $I_\tau$ overlap, hence their union $I'_v$ spans a connected subtree as well. \qedhere
\end{description}
\end{proof}

\subparagraph*{Algorithmic aspects of \Cref{thm:heegaard-diag-tw}} 
We briefly comment on the computational cost of extracting the Heegaard diagram $\diag(\tri)$ from a given triangulation $\tri$ with $n$ tetrahedra. Recall that for an \emph{abstract simplicial complex} $\complex$, its barycentric subdivision $\complex'$ is the simplicial complex whose ground set is $\complex$, and a collection $\{\tau_0,\ldots,\tau_k\}$ forms a $k$-simplex in $\complex'$ precisely when the simplices $\tau_i \in \complex$ are nested, that is, when $\tau_0 \subset \cdots \subset \tau_k$. Although a triangulation $\tri$ is generally not a simplicial complex, we can still form its barycentric subdivision by first subdividing each tetrahedron and then applying the gluing maps to identify the subdivided triangular faces. Performed twice, the resulting complex $\tri''$ naturally contains both the canonical handle decomposition $\chd(\tri)$ and the induced Heegaard diagram $\diag(\tri)$. As the barycentric subdivision of a tetrahedron contains 24 tetrahedra, $|\tri''| = 24^2|\tri| = 576|\tri|$.

\proofsubparagraph{Extracting chd($\bm{\protect\tri}$) from $\bm{\protect\tri''}$} Every vertex $v$ of $\tri'$ is of the form $v = \{\tau\}$ for some simplex $\tau \in \tri$. Let $\iota\colon \tri'(0) \hookrightarrow \tri''(0)$ be the natural inclusion defined by $\iota(\{\tau\}) = \{\{\tau\}\}$. For every $w \in \operatorname{im}(\iota)$, let $h_w$ be the smallest subcomplex of $\tri''$ containing all $3$-faces incident to $w$. If $w$ corresponds to a simplex $\sigma \in \tri$ of dimension $3-k$, then $h_w$ is regarded as a $k$-handle. The collection $\{h_w : w \in \operatorname{im}(\iota)\}$ forms the canonical handle decomposition $\chd(\tri)$. See \Cref{fig:barycentric}.

\begin{figure}[ht]
	\centering
	\begin{minipage}[t]{.24\textwidth}
		\centering
		\includegraphics[scale=.8,page=1]{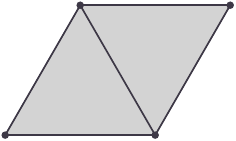}
		
		\small{$\tri$}~~~~~
		
	\end{minipage}\hfill%
	\begin{minipage}[t]{.24\textwidth}
		
		\centering
		\includegraphics[scale=.8,page=2]{barycentric-alt}
		
		\small{$\tri'$}~~~~~

	\end{minipage}\hfill%
	\begin{minipage}[t]{.24\textwidth}
		\centering
		\includegraphics[scale=.8,page=3]{barycentric-alt}
		
		\small{$\tri''$}~~~~~

	\end{minipage}
	\begin{minipage}[t]{.24\textwidth}
		\centering
		\includegraphics[scale=.8,page=4]{barycentric-colored-alt-2}
		
		\small{$\chd(\tri)$}~~~~~

	\end{minipage}
	
	\caption{A 2-dimensional triangulation $\tri$ with two triangles, its first and second barycentric subdivisions $\tri'$ and $\tri''$ (regarded geometrically), and the canonical handle decomposition $\chd(\tri)$.}
	\label{fig:barycentric}
\end{figure}

\proofsubparagraph{Extracting $\bm{\protect\diag(\protect\tri)}$ from $\bm{\protect\tri''}$} To extract the induced Heegaard diagram $\diag(\tri)$ from $\tri''$, we need to pinpoint its $\alpha$- and $\beta$-curves, as well as the splitting surface $\surface$ in $\tri''$. We illustrate this via concrete examples; see \Cref{fig:locating}. We start with the $\alpha$-curves. Let $e \in \tri(1)$ be an edge of $\tri$, and assume that $e$ is contained in four distinct tetrahedra $\sigma_1,\sigma_2.\sigma_3,\sigma_4$ of $\tri$. Let $t_{ij}$ denote the triangular face shared by $\sigma_i$ and $\sigma_j$. Consider the subcomplex $C_e$ of $\tri'$ spanned by the nine vertices $\{e\}, \{\sigma_1\}, \{\sigma_2\}, \{\sigma_3\}, \{\sigma_4\}, \{t_{12}\}, \{t_{23}\}, \{t_{34}\}$, and $\{t_{14}\}$. This complex $C_e$ triangulates the \emph{dual $2$-cell transverse to $e$} (\Cref{fig:locating}, top left). Taking the second barycentric subdivision $\tri''$ (\Cref{fig:locating}, top middle), we see that the $\alpha$-curve $\alpha_e$ corresponding to $e$ is readily obtained as the polygonal curve spanned by the vertices of $C'_e$ adjacent to $\{\{e\}\}$ (\Cref{fig:locating}, top right). The $\beta$-curves can be located analogously (\Cref{fig:locating}, bottom row).

\begin{figure}[ht]
	\begin{minipage}[t]{1\textwidth}
	\centering
	\includegraphics[scale=1]{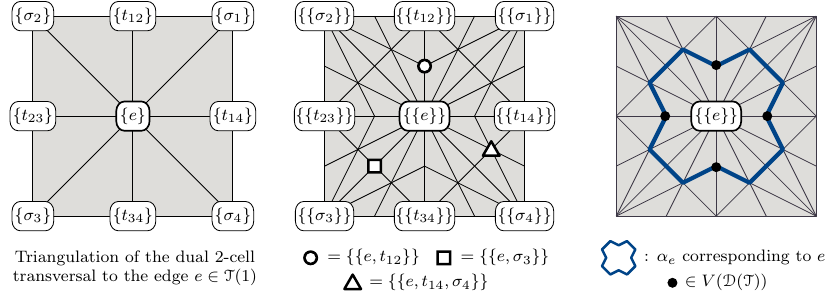}
	
	\end{minipage}
	
	\vspace{1em}

	\noindent
	\begin{minipage}[t]{1\textwidth}
	\centering
	\includegraphics[scale=1]{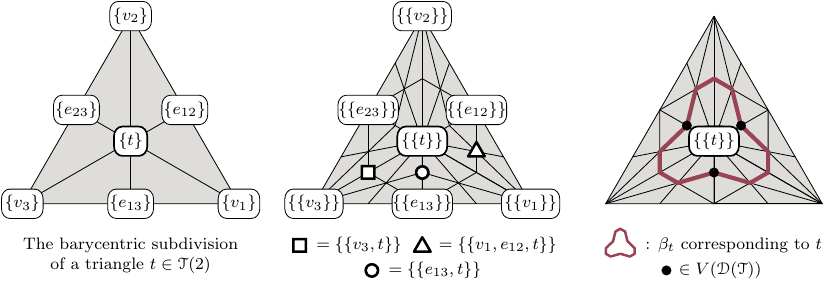}
	
	\end{minipage}
	
	\caption{Locating the $\alpha$- and $\beta$-curves $\diag(\tri)$ in the second barycentric subdivision $\tri''$ of $\tri$.}
	\label{fig:locating}
\end{figure}

A triangulation of the splitting surface $\surface$ can also be readily obtained as the boundary of the subcomplex of $\tri''$ formed by the union of the triangulated 2- and 3-handles of $\chd(\tri)$. Using standard data structures to represent these objects (such as in \texttt{Regina} \cite{regina}), all of these operations can be performed in linear time in terms of $n = |\tri|$.

\begin{remark}
\label{rem:min-heegaard}
Often it is assumed—and in certain applications even required, cf.\ \cite{kuperberg1996noninvolutory}—that a Heegaard diagram $\diag = (\surface,\alphacurves,\betacurves)$ be \emph{minimal}, i.e.\ $|\alphacurves| = |\betacurves| = g(\surface)$. In our setting, we can ensure this by first choosing arbitrary spanning trees $T_1$ and $T_2$ for the 1-skeleton $\tri^{(1)}$ and the dual graph $\dual(\tri)$ of $\tri$, respectively, and then retaining only those $\alpha$-curves (resp.\ $\beta$-curves) that correspond to edges of $\tri^{(1)}$ (resp.\ $\dual(\tri)$) \emph{not} contained in $T_1$ (resp.\ $T_2$).
\end{remark}

\section{Retriangulating with constant edge valence and controlled treewidth}
\label{sec:retri}


In this section, we describe and analyze a quasi-linear time procedure that retriangulates any 3-manifold triangulation $\tri$ into a triangulation $\tri^\ast$ with maximum edge valence at most $9$ and treewidth increased only by a polylogarithmic factor of the original maximum edge valence.

\begin{theorem}[\Cref{thm:generalalgo-intro}; with explicit polylog-factors]
\label{thm:generalalgo}
There is an algorithm which, given a triangulation $\tri$ of a closed $3$-manifold $\M$ with $n$ tetrahedra, $\vertices$ vertices, $\tw(\dual(\tri)) = \treewidth$, and maximum edge valence $\maxval(\tri) = \maxvalence$, yields a triangulation $\tri^\ast$ of $\M$ with $O((n+\vertices)\log_2(\maxvalence)^{5.32})$ vertices and tetrahedra, $\tw(\dual(\tri^\ast)) < (\treewidth + 1) \cdot \log_2(\maxvalence)^{5.17}$, and $\maxval(\tri^\ast) \leq 9$.
The algorithm runs in $O(n \log\log n)$ time.
\end{theorem}

The construction in \Cref{thm:generalalgo} is based on an iterative application of the following result.

\begin{theorem}
There is an algorithm which, given a triangulation $\tri$ of a closed $3$-manifold $\M$ with $n$ tetrahedra, $\vertices$ vertices, $\tw(\dual(\tri)) = \treewidth$, and $\maxval(\tri) = \maxvalence$, constructs a triangulation $\tri^\ast$ of $\M$ with at most $(28+4\sqrt{6})n+(16+4\sqrt{6})\vertices$ tetrahedra, $(6+\sqrt{6})(n+\vertices)$ vertices, $\maxval(\tri^\ast) \leq \max\{\lfloor\sqrt{\maxvalence}\rfloor+4,9\}$, and $\tw(\dual(\tri^\ast)) \leq 36 \cdot \treewidth$. The algorithm runs in $O(n)$ time.
\label{thm:itestep}
\end{theorem}

The proof of \Cref{thm:itestep} is the core of this section. We prove \Cref{thm:generalalgo} at the end of it.

\subparagraph*{Setup} Let us fix a triangulation $\tri$ of a closed 3-manifold $\M$ with $n$ tetrahedra, $|\tri(0)| = \vertices$ vertices, treewidth $\tw(\dual(\tri)) = \treewidth$, and maximum edge valence $\maxval(\tri) = \maxvalence$. Let $\spine = \spine(\tri)$ be the \emph{spine of $\M$ dual to $\tri$}. It is a $2$-dimensional polyhedral cell complex whose $1$-skeleton $\spine^{(1)}$ is the dual graph $\dual(\tri)$, and for each edge $e \in \tri(1)$, $\spine$ contains a polygonal $2$-face $\face_e$ glued along the cycle in $\dual(\tri)$ formed by the vertices corresponding to the tetrahedra that contain~$e$. Locally, at each tetrahedron of $\tri$, the spine $\spine$ looks like the configuration shown in Figure~\ref{fig:spine}.

\begin{remark*}
Spines, and in particular \emph{special spines}, offer a perspective on 3-manifolds that is dual to that provided by triangulations. We refer to \cite[Chapter~1]{matveev2007algorithmic} and \cite{rubinstein2019traversing} for details.
\end{remark*}

\subparagraph*{Overall strategy} Note that, given the spine $\spine$, we can recover $\manifold$ by first thickening $\spine$ and then filling in the resulting $2$-sphere boundary components with $3$-balls (which may be viewed as centered at the vertices of $\tri$). In what follows, we triangulate $\spine$ and fill in these missing balls combinatorially. Throughout, the term \emph{$2$-face} refers to the faces of the spine $S$, whereas \emph{triangle} refers to the faces of the triangulation. Our strategy is similar to that of~\cite[Theorem 3.36]{frick2015thesis}, but we employ a different---and necessarily more intricate---triangulation of the $2$-faces of $\spine$ in order to control the treewidth of the resulting triangulation.

\begin{figure}[ht]
\centering
\includegraphics[width=5.5cm]{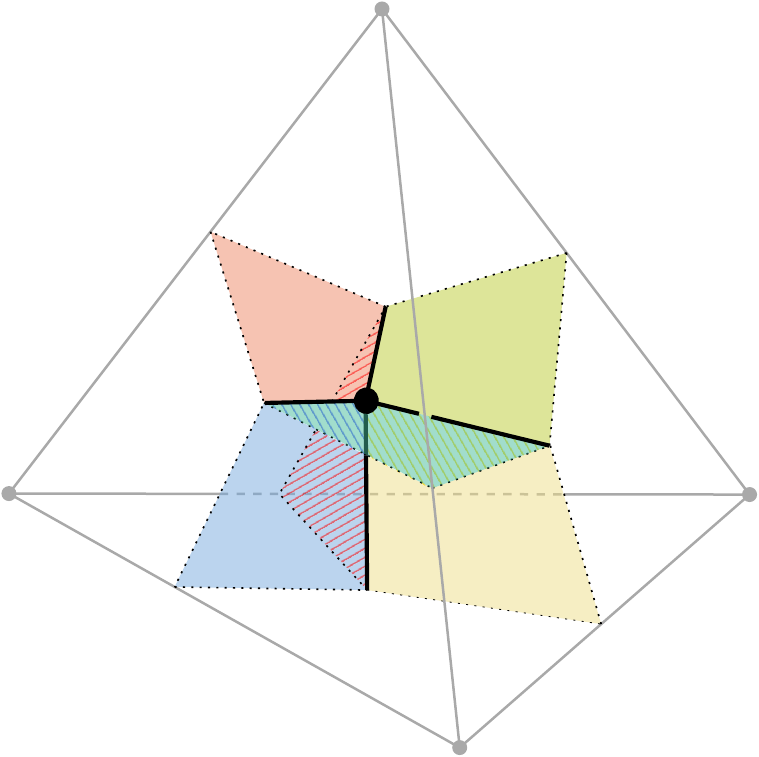}
\caption{The intersection of the spine $\spine$ with a tetrahedron of $\tri$. The vertices of $\spine$ bijectively correspond to the centers of tetrahedra of $\tri$, and each $2$-face of $\spine$ is transverse to an edge of $\tri$.}
\label{fig:spine}
\end{figure}

\subparagraph*{Triangulating the faces of $\bm\spine$} 
Consider a $2$-face $\face \defeq \face_e$ of the spine $S$ transverse to an edge $e \in \tri(1)$. This face $\face$ is a $k$-gon, where $k = \val(e)$, whose vertices are the centers of the tetrahedra of $\tri$ containing $e$. We triangulate $\face$ as follows (see Figure~\ref{fig:triangulatewing}).
If $k \leq 9$, we add a vertex $w$ at the center of $\face$ and connect it to all boundary vertices, as in \Cref{fig:triangulatewing} (left).
If $k > 9$, we proceed as follows. Fix an arbitrary orientation of the boundary of $\face$, and choose positive integers $m \defeq m_\face$ and $d \defeq d_\face$ such that $k \leq m(d-1)$. Label the boundary vertices $u_0, \ldots, u_{k-1}$ in accordance with this orientation. Introduce a ring of $m$ interior vertices $v_0, \ldots, v_{m-1}$, and connect each $v_i$ to $v_{i-1}$ and $v_{i+1}$ (indices $\mod m$). For each $0 \leq i < m-1$, connect $v_i$ to the boundary vertices $u_{i(d-1)}, \ldots, u_{(i+1)(d-1)}$. The final interior vertex $v_{m-1}$ is connected to $u_{(m-1)(d-1)}, \ldots, u_{k-1}, u_0$. The inequalities $(m-1)(d-1) < k \leq m(d-1)$ ensure that the vertices $u_i$ and $v_\ell$, together with these edges, form a triangulated annulus. Finally, we place a central vertex $w$ inside $\face$ and connect it to all $v_i$, see \Cref{fig:triangulatewing} (right).

The next proposition follows directly from this construction; $k$, $m$, and $d$ are as above.

\begin{proposition}
The triangulation of the $2$-face $\face$ described above consists of $k + 2m$ triangles. Each boundary vertex $u_i$ is incident to $3$ or $4$ edges, each interior vertex $v_j$ is incident to at most $d+3$ edges, and the central vertex $w$ is incident to $m$ edges.
In particular, the maximum degree of a vertex in the $1$-skeleton of $\face$ is $\max\{m,d+3\}$.
\label{prop:maxdeg}
\end{proposition}

\begin{figure}[ht]
\centering
\includegraphics[width=1\textwidth]{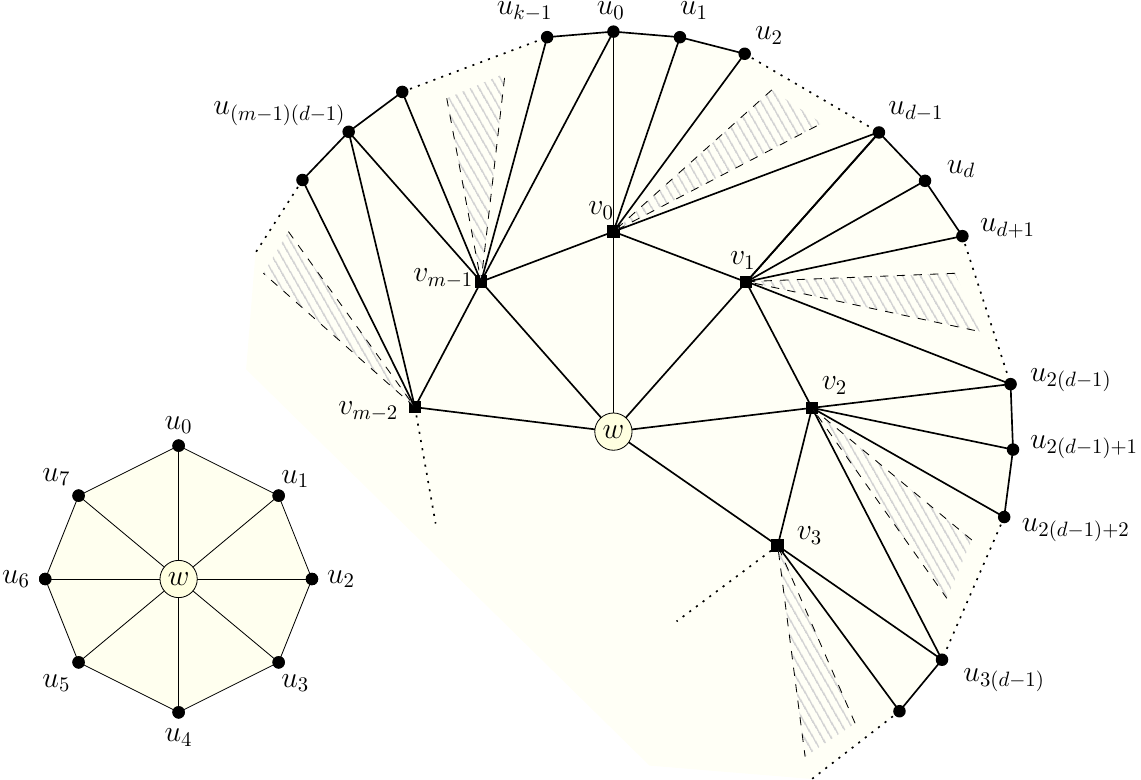}
\caption{Triangulating a $2$-face $\face$ (a $k$-gon) of the spine $\spine$ of $\tri$ if $k \leq 9$ (left) and if $k > 9$ (right).}
\label{fig:triangulatewing}
\end{figure}

\subparagraph*{Building the triangulation $\protect\bm{\protect\tri^\ast}$ in \Cref{thm:itestep}}
Let $\tri_{\spine}$ denote the triangulation of the spine $\spine$ obtained by triangulating each $2$-face $\face$ of $\spine$ as described above. We construct the triangulation $\tri^\ast$ by first taking $\tri_{\spine}$, considered as embedded in the geometric realization $\|\tri\|$ of $\tri$, and then adding the vertices $\tri(0)$. Each vertex $a \in \tri(0)$ lies in the interior of a connected component of $\|\tri\| \setminus \spine$, which is an open $3$-ball denoted by $B_a$. Its boundary $\partial B_a := \overline{B_a} \setminus B_a$ is a subpolyhedron of $\spine$ comprised of certain $2$-faces of $\spine$. Let $\tri_{\partial B_a}$ denote the subtriangulation of $\tri_{\spine}$ formed by the union of the triangulations of these $2$-faces of $\spine$.  
To complete the construction, for each $a \in \tri(0)$ we cone over $\tri_{\partial B_a}$ from the vertex $a$. We denote the resulting triangulation by $\tri^\ast$. By construction, $\tri^\ast$ triangulates $\M$.

\begin{remark*}
In the above construction, for each $2$-face $\face$ we fix integers $m = m_{\face}$ and $d = d_{\face}$ that depend solely on the number of boundary vertices $u_i$ of $\face$ in $\spine$, or equivalently, on the valence of the edge of $\tri$ transverse to $\face$. When the context requires, we write $d_{\face}$ and $m_{\face}$.
\end{remark*}

\begin{definition}[type of a vertex of {$\tri^\ast$}]
The vertices of $\tri^\ast$ fall into the following three \emph{types}. \begin{enumerate*}
	\item vertices $a \in \tri(0)$,
	\item vertices $u_i$ lying on the boundary of exactly six $2$-faces of $\spine$, and
	\item vertices $v_\ell$ and $w$ lying in the interior of exactly one $2$-face of $\spine$.
\end{enumerate*} 
\end{definition}

We proceed with investigating the maximum edge valence and the treewidth of $\tri^\ast$.

\begin{lemma}
The triangulation $\tri^\ast$ has edge valences at most $\max_{F \in S} \max\{d_F+3,m_F,9\}$.
\label{lem:valence}
\end{lemma}

\begin{proof}
We consider, case by case, the different types of edges in $\tri^\ast$. Let $a$ denote a vertex of $\tri(0)$, and let $u_i$, $v_\ell$, and $w$ denote vertices in a $2$-face $F$, see \Cref{fig:triangulatewing}. (We keep this labeling convention throughout this section.) The next valence-bounds follow from \Cref{prop:maxdeg}.

\begin{description}
\item[\emph{Type $\{a,w\}$}] The valence of the edge of this type is equal to the degree of $w$ in the triangulation of the 2-face, i.e.,\ $\val(\{a,w\}) = m$.
\item[\emph{Type $\{a,v_\ell\}$}] The valence of the edge of this type is equal to the degree of $v_\ell$ in the triangulation of the 2-face, i.e.,\ $d+3$.
\item[\emph{Type $\{a,u_i\}$}] is incident to at most $3$ tetrahedra per 2-face containing $u_i$, and there are at most three 2-faces containing $u_i$ and bounding the ball containing $a$, so $\val(\{a,u_i\}) \leq 9$.
\item[\emph{Type $\{w,v_\ell\}$ or $\{v_\ell,u_i\}$}] are incident to at most $4$ tetrahedra, two on either side of $F$,
\item[\emph{Type $\{u_i,u_{i+1}\}$}] is incident to one tetrahedra $\{a,v_\ell,u_i,u_{i+1}\}$ per 2-face that contains it, and is contained in at most three 2-faces, therefore $\val(\{u_i,u_{i+1}\}) \leq 3$.
\end{description}

Hence, for the maximum edge valence of $\tri^\ast$ we have $\maxval(\tri^\ast) \leq \max\{d_F+3,m_F,9\}$. 
\end{proof}

\begin{lemma}
The dual graph of the triangulation $\tri^\ast$ has treewidth $\tw(\dual(\tri^\ast)) < 36\cdot(\treewidth +1)$. 
\label{lem:36tw}
\end{lemma}

\begin{proof}
\pushQED{}
Consider a tree decomposition $\treedecomp = (T,\bags)$ of the graph $\dual(\tri)$ with optimal width $\treewidth$.

By construction, every tetrahedron of $\tri^\ast$ contains a vertex of $\tri(0)$ and at least one triangle from the triangulation of a $2$-face $\face$ of $\spine$. Following the notation of Figure~\ref{fig:triangulatewing}, these triangles are of the form $\{v_\ell, u_{i-1}, u_i\}$, $\{v_{\ell-1}, u_{\ell(d-1)}, v_{\ell+1}\}$, or $\{w, v_{\ell-1}, v_\ell\}$.
For any 2-face of $\spine$, the boundary vertices $u_0, \ldots, u_{k-1}$ correspond to nodes of the graph $\dual(\tri)$, and consequently each of them appears in some bag of this tree decomposition $\treedecomp$.

We now construct a tree decomposition $\treedecomp'=(T,\bags')$ for the dual graph of $\tri^\ast$, with the same underlying tree $T$ as for $\treedecomp$ and $\bags' = \{B'_\tau : \tau \in V(T)\}$. We initialize $B'_\tau$ as $B'_\tau = \emptyset$ for all $\tau \in V(T)$, and consider all types of tetrahedra $\sigma$ in $\tri^\ast$:
\begin{enumerate}
\item\label{itm:treedec1} If $\sigma$ has a triangular face $\{v_\ell, u_{i-1}, u_{i}\}$, then add $\sigma$ to all bags $B'_\tau$ such that $u_{i-1} \in B_\tau$.

\item\label{itm:treedec2} If $\sigma$ has a face $\{v_{\ell-1},u_{\ell(d-1)},v_{\ell+1}\}$, then add $\sigma$ to all bags $B'_\tau$ such that $u_{\ell(d-1)} \in B_\tau$. Similarly, for $\sigma = \{a,v_{m-1},u_0,v_0\}$, add $\sigma$ to all bags $B'_\tau$ such that $u_{0} \in B_\tau$.

\item\label{itm:treedec3} If $\sigma$ has a face $\{w,v_{\ell-1},v_{\ell}\}$, then add $\sigma$ to all bags $B'_\tau$ such that $B_\tau$ contains any $u_{i}$ with $\ell(d-1) \leq i < \max \{(\ell+1)(d-1),k-1\}$.
\end{enumerate}

\begin{figure}[ht]
\centering
\includegraphics[scale=.875]{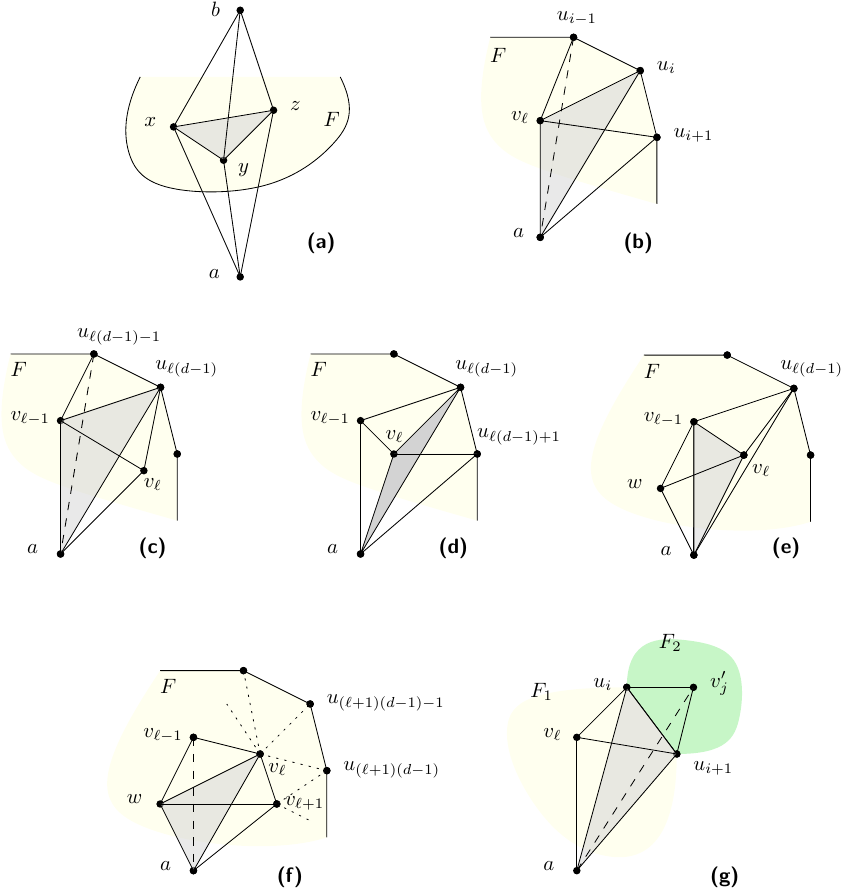}

\medskip
\caption{Illustrations of 7 possible triangles shared by tetrahedra in $\tri^\ast$ from the proof of \Cref{claim:decomp}.}
\label{fig:retri_cases}
\end{figure}

\begin{claim}
\label{claim:valid_tree_dec_retri}
$\treedecomp' = (T,\bags')$ is a valid tree decomposition of the dual graph $\dual(\tri^\ast)$ of $\tri^\ast$.
\label{claim:decomp}
\end{claim}

\begin{claimproof}
We verify the three defining properties of a tree decomposition (Section~\ref{ssec:graphs}). 
\begin{description}
\item[\emph{Vertex coverage}] By construction, any tetrahedron of $\tri^\ast$ contains a triangle from a 2-face $F$ of the spine $\spine$, and in consequence appears in some bag of $\treedecomp' = (T,\bags')$. 
\item[\emph{Edge coverage}] Showing that every edge of $\tri^\ast$ is contained in some bag of $\treedecomp'$ can be carried out by a tedious case analysis. In the triangulation $\tri^\ast$, two tetrahedra may be adjacent along seven types of shared triangles, classified by the types of the triangle’s vertices; all possible configurations are shown in Figure~\ref{fig:retri_cases}. See Appendix~\ref{app:missingproof} for all the details.

\item[\emph{Subtree property}] By~\myref{itm:treedec1} and~\myref{itm:treedec2}, tetrahedra of type $\{a,v_\ell,u_i,u_{i+1}\}$, $\{a,v_{\ell-1},u_{\ell(d-1)},v_\ell\}$ and $\{a,v_{m-1},u_0,v_0\}$ appear exactly in the bags $B'_\tau$ such that $B_\tau$ contain a certain fixed tetrahedron (respectively $u_i$, $u_{\ell(d-1)}$, and $u_0$); the bags containing them are consequently connected in $\sigma$. By~\myref{itm:treedec3}, a tetrahedron $\sigma$ of type $\{a,w,v_{\ell-1},v_\ell\}$ is inserted in all bags $B'_\tau$ such that $B_\tau$ contains any $u_{i}$ with $\ell(d-1) \leq i < (\ell+1)(d-1)$. Call $T_i$ the subtree of $T$ spanned by the nodes $\tau$ such that $u_i \in B_\tau$ in $\treedecomp$. The subtree of $T$ spanned by the nodes $\tau$ in $\treedecomp'$ such that $\sigma \in B_\tau$ is, in consequence, $\cup_{\ell(d-1) \leq i < (\ell+1)(d-1)} T_i$. Because $\treedecomp$ is a tree decomposition, each $T_i$ is connected. Additionally, for any $i$, $u_i$ and $u_{i+1}$ are adjacent, and must consequently appear in a common bag of $\treedecomp$. Hence $T_i \cap T_{i+1} \neq \emptyset$. In consequence, $\cup_{\ell(d-1) \leq i < (\ell+1)(d-1)} T_i$ is a connected subtree of $\sigma$. \claimqedhere
\end{description}
\end{claimproof}

\begin{claim}
$\treedecomp' = (T,\bags')$ has width less than $36\cdot(\treewidth +1)$.
\label{claim:width}
\end{claim}

\begin{claimproof}
Let $B_\tau$ be a bag of $\treedecomp$ of size $r$. Every occurrence of a tetrahedron $u_i$ in $B_\tau$ will induce the insertion of at most $6$ tetrahedra in $B'_\tau$, per 2-face incident to $u_i$, namely tetrahedra $(\star, v_\ell, u_i, u_{i+1}\}$, tetrahedra $(\star,v_\ell,u_{i},v_{\ell+1}\}$ if $i=\ell(d-1)$ or $i=0$, and tetrahedra $(\star,w,v_{\ell-1},v_{\ell}\}$ with $\ell(d-1) \leq i \leq \max \{(\ell+1)(d-1)-1,k-1\}$; where $\star$ indicates at most two possible vertices $a$ and $b$ of $\tri(0)$, on either side of the 2-face. Adding the fact that at most six 2-faces are incident to a given tetrahedron $u_i$, the size of a bag $B'_\tau$ is at most $36$ times the size of $B_\tau$. In consequence, the width of $\treedecomp'$ is at most $36\cdot(\treewidth +1)-1$. This concludes the proofs of \Cref{claim:width} and \Cref{lem:36tw}. \qedhere
\end{claimproof}
\end{proof}

From now on, for the $2$-face $F_e$ transverse to the edge $e \in \tri(1)$ with valence $k = \val(e)$, we set $m \defeq \lfloor \sqrt{\val(e)} \rfloor + 4$ and $d \defeq \lfloor \sqrt{\val(e)} \rfloor + 1$, which satisfy the condition $k \leq m(d-1)$ whenever $\val(e) > 6$. The next two lemmas provide bounds on the number of tetrahedra (Lemma~\ref{lem:numtet}) and the number of vertices (Lemma~\ref{lem:vertices}) in $\tri^\ast$. These rely on standard counting arguments for triangulations; see Appendix~\ref{ssec:countval}.

\begin{lemma}
The triangulation $\tri^\ast$ has at most $(28+4\sqrt{6})n+(16+4\sqrt{6})\vertices$ tetrahedra.
\label{lem:numtet}
\end{lemma}

\begin{proof}
Each 2-face, transversal to an edge $e$ of valence $\val(e)$, contributes to $2(k+2m) = 2\val(e) + 4(\lfloor\sqrt{\val(e)}\rfloor + 4)$ tetrahedra ($k+2m$ triangles in the transversal 2-face $F$, holding a tetrahedron on either side of $F$). In consequence,
\begin{align*}
  \sum_{e \in \tri(1)}(2\val(e) + 4(\lfloor\sqrt{\val(e)}\rfloor + 4)) \underbrace{\leq}_{\text{Lem.~\ref{lem:eulerchi}, Cor.~\ref{lem:sumval}}} 12n + 16(n+\vertices) + 4\sum_e \sqrt{\val(e)} \\ 
\underbrace{\leq}_{\text{Lem.~\ref{lem:sumvalsqrt}}} 12n + 16(n+\vertices) + 4\sqrt{6}(n+\vertices).  \tag*{\qedhere}
\end{align*}
\end{proof}

\begin{lemma}
The triangulation $\tri^\ast$ has at most $(6+\sqrt{6})(n+\vertices)$ vertices.
\label{lem:vertices}
\end{lemma}

\begin{proof}
By construction, the vertex set $\tri^\ast(0)$ consist of the vertices $\tri(0)$, the vertices $u_i$ (lying on the boundaries of the $2$-faces), and, for each $e \in \tri(1)$, the $\lfloor \sqrt{\val(e)} \rfloor + 5$ interior vertices on the triangulated $2$-face $\face_e$. The vertices $u_i$ are in bijection with the tetrahedra of $\tri$, hence there are exactly $n$ of them. The total number of interior vertices over all $2$-faces is
\[
\sum_{e \in \tri(1)} \bigl(\lfloor \sqrt{\val(e)} \rfloor + 5\bigr)
\;\; \underbrace{\leq}_{\text{Lemmas~\ref{lem:eulerchi}, \ref{lem:sumvalsqrt}}} \;\;
(\sqrt{6}+5)(n+\vertices).
\]
With $|\tri(0)| = \vertices$ and the $n$ vertices $u_i$, this gives at most $(6+\sqrt{6})(n+\vertices)$ vertices in $\tri^\ast$.
\end{proof}

We are now ready to prove the main theorems of this section.

\begin{proof}[Proof of Theorem~\ref{thm:itestep}]
If the maximum edge valence of $\tri$ is at most $9$, the algorithm simply returns $\tri^\ast \defeq \tri$. Otherwise, we apply the construction described above; the bounds on the number of tetrahedra, the number of vertices, the treewidth, and the edge valence follow from Lemmas~\ref{lem:numtet}, \ref{lem:vertices}, \ref{lem:36tw}, and \ref{lem:valence}, respectively. Note that the construction runs in time linear in the size of $\tri^\ast$, which is $O(n)$.
\end{proof}

Finally, the proof of Theorem~\ref{thm:generalalgo} relies on standard calculations involving geometric series and iterated radicals, as detailed in Lemmas~\ref{lem:comp1}, \ref{lem:limit}, and \ref{lem:doublesequence} in Appendix~\ref{ssec:iterated}.

\begin{proof}[Proof of Theorem~\ref{thm:generalalgo}]
If the maximum edge valence of $\tri$ satisfies $\maxval(\tri) \leq 9$, we are done. Otherwise, we construct a sequence of triangulations $\tri = \tri_0, \tri_1, \ldots, \tri_p$ by iterating the construction of Theorem~\ref{thm:itestep}. The triangulation $\tri_r$ has $O((30+4\sqrt{6})^r (n+\vertices))$ tetrahedra and vertices (Lemma~\ref{lem:doublesequence}), treewidth $\tw(\dual(\tri_r)) < 36^r \cdot(\treewidth +1)$ (Lemma~\ref{lem:36tw}), and maximum edge valence $\maxval(\tri_r) \leq \max\{a_r,9\}$ (Lemma~\ref{lem:valence}, with $m=\lfloor \sqrt{\val(e)} \rfloor + 4$ and $d=\lfloor \sqrt{\val(e)} \rfloor + 1$), where $(a_r)_{r \in \mathbb{N}}$ is the sequence of iterated radicals of Lemma~\ref{lem:comp1}, with $a_0=\maxvalence$, $c=4$, and $L<6.32$. By Lemma~\ref{lem:limit}, choosing $p=\log_2\log_2(\maxvalence)$ ensures that the edge valence of $\tri_p$ is at most $9$. We therefore set $\tri^\ast := \tri_p$. Since $\log_2(30+4\sqrt{6})<5.32$ and $\log_2(36)< 5.17$, the triangulation $\tri^\ast$ has $O((n+\vertices)\log_2(\maxvalence)^{5.32})$ tetrahedra and vertices, treewidth at most $(\treewidth+1)\cdot\log_2(\maxvalence)^{5.17}$, and edge valence at most $9$. The algorithm performs $O(\log\log(n))$ iterations of the algorithm of Theorem~\ref{thm:itestep}, yielding the stated complexity bound.
\end{proof}

\section{Efficient evaluation of Kuperberg's tensor networks}
\label{sec:tensor}


\noindent In this section we present a fixed-parameter tractable (FPT) scheme for computing Kuperberg's quantum invariants $\sharp(\odiag,\hopf)$ of closed, oriented\footnote{In this section we assume $\manifold$ to be an \emph{oriented}, not merely orientable, 3-manifold.} 3-manifolds \cite{kuperberg1991involutory}. This quantity depends on an oriented Heegaard diagram\footnote{An \emph{oriented Heegaard diagram} is a Heegaard diagram $\diag = (\surface,\alphacurves,\betacurves)$ together with an orientation for the surface $\surface$ and for each of the $\alpha$- and $\beta$-curves. The quantity $\sharp(\odiag,\hopf)$ does not depend on the choice of orientation for the $\alpha$- and $\beta$-curves, however, it generally does depend on the chosen orientation for $\surface$.} $\odiag = (\osurface,\oalpha,\obeta)$ and a finite-dimensional involutory Hopf algebra $\hopf$ over a field $\field$, which is considered fixed. By construction, $\sharp(\odiag,\hopf) \in \field$. 

\begin{remark*}
We emphasize that understanding the results in this section does \textbf{not} require knowledge of Hopf algebras or tensors, since the following proofs are purely graph theoretical. Nevertheless, we provide some basic definitions and examples in Appendix~\ref{app:tensor}. We refer to Kuperberg's paper \cite{kuperberg1991involutory} (cf.\ its open-access version \cite{kuperberg1991involutory-arxiv}) for further details.
\end{remark*}

Even if $\sharp(\odiag,\hopf)$ is originally derived from an oriented Heegaard diagram, it can also be computed from an oriented triangulation. This is briefly hinted at in \cite[Section~4]{kuperberg1991involutory}, however, with the machinery developed in \Cref{sec:algo} we can make it fully explicit: Any triangulation $\tri$ of a 3-manifold $\manifold$ gives rise to a Heegaard diagram $\diag = (\surface,\alphacurves,\betacurves)$ of $\manifold$, moreover, in a width-preserving way (\Cref{thm:heegaard-diag-tw}). Picking an orientation for $\tri$ induces an orientation also for $\surface$. Then choosing orientations for the $\alpha$- and $\beta$-curves arbitrarily yields an oriented diagram $\odiag$. Next, from $\odiag$ we build a \emph{tensor network} $\network$ following Kuperberg's construction \cite[Section~5]{kuperberg1991involutory}, which we finally \emph{evaluate} to obtain $\sharp(\odiag,\hopf)$.
From now on, we will use the notation $\sharp(\otri,\hopf) \defeq \sharp(\odiag,\hopf)$ to indicate when we regard the invariant as computed from a triangulation, where $\odiag$ denotes the oriented Heegaard diagram induced by $\otri$.

Building on \Cref{thm:heegaard-diag-tw} and bounding the treewidth of the tensor network $\network$ (\Cref{lem:kuperberg-width}), together with work of O'Gorman \cite{ogorman2019parameterization} (cf.\ Markov--Shi \cite{markov2008simulating}) we prove the following result:

\begin{theorem}
\label{thm:kuperberg-fpt}
Let $\hopf$ be a fixed finite-dimensional involutory Hopf algebra over a field~$\field$. Let $\otri$ be an oriented triangulation of a closed oriented $3$-manifold with $n$ tetrahedra and treewidth $\tw(\otri)\defeq\tw(\dual(\otri))=t$. The invariant $\sharp(\otri,\hopf)$ can be computed in time $O(2^{O(t)}n)$.
\end{theorem}

Before discussing the proof of \Cref{thm:kuperberg-fpt} in detail, we provide an overview on \Cref{fig:pipeline}.

\begin{figure}[ht]
	\centering
	\vspace{1.75cm}
	\begin{overpic}[scale=1]{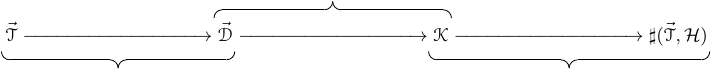}
		\put (-3.75,-7.5) {%
			\fbox{%
				 \parbox{1.1cm}{%
				 	\footnotesize{\,{Thm.\ \ref{thm:heegaard-diag-tw}}}}
				\parbox{3.5cm}{%
					\footnotesize{%
						\begin{itemize}
							\itemsep0em
							\item $|V(\odiag)| = O(n\defeq|\otri|)$
							\item $\tw(\odiag) = O(\tw(\dual(\otri)))$
						\end{itemize}
						\vspace{-.5\topsep}
					}%
				}%
			}%
		}
		\put (28.5,16) {%
			\fbox{%
				 \parbox{1.1cm}{%
				 	\footnotesize{\,{Lem.\ \ref{lem:kuperberg-width}}}}
				\parbox{3.1cm}{%
					\footnotesize{%
						\vspace{-.5\topsep}
						\begin{itemize}
							\itemsep0em
							\item$|V(\network)| = 3|V(\odiag)|$
							\item$\tw(\network) = 2\tw(\odiag)$
						\end{itemize}%
						\vspace{-.25\topsep}
					}%
				}%
			}%
		}
		\put (60.75,-7.5) {%
			\fbox{%
				\parbox{4.4cm}{%
					\centering\footnotesize{%
						\vphantom{$\vec{M}$}Cor.\ \ref{cor:contraction}:\ $\network$ can be evaluated in\\[.5em]
						$O\big(2^{O(\tw(\otri))}n\big)~\text{time}$
					}%
				}%
			}%
		}
		\put (37,1.5) {\footnotesize{\textbf{\textsf{\color{lipicsGray}2.}}\ \cite[Section~5]{kuperberg1991involutory}}}
		\put (9,8) {\parbox{1.75cm}{\centering\footnotesize{\textbf{\textsf{\color{lipicsGray}1.}}\ \Cref{sec:algo}\\ (this paper)}}}
		\put (66.75,8) {\parbox{2.5cm}{\centering\footnotesize{\textbf{\textsf{\color{lipicsGray}3.}}\ \cite[Thm.\ 1]{ogorman2019parameterization}\\ (cf.~\cite[Thm.\ 4.6]{markov2008simulating})}}}
	\end{overpic}
	\vspace{1.5cm}

	\caption{The algorithmic pipeline of computing $ \sharp(\otri,\hopf)$ from an oriented triangulation $\otri$. The references on the arrows provide the theoretical backbone, and the framed results indicate our contributions. The running time of the computation is dominated by that of step \textbf{\textsf{\color{lipicsGray}3.}}\ $\network \longrightarrow \sharp(\otri,\hopf)$.}
	\label{fig:pipeline}
\end{figure}

\begin{remark}
\Cref{thm:kuperberg-fpt} asserts that computing $\sharp(\otri,\hopf)$  is FPT in the treewidth of $\otri$. In particular, for bounded-treewidth triangulations it can be computed in linear time. However, since tensor network contraction---the last part of the computation---is generally $\sharp\mathbf{P}$-hard \cite{damm2002complexity}, we expect this to hold for computing $\sharp(\otri,\hopf)$ for general triangulations as well.
\end{remark}

As motivation, let us recall the formula for the invariant $\sharp(\odiag,\hopf)$ from \cite[Section~5]{kuperberg1991involutory}:
\begin{align}
	\sharp(\odiag,\hopf) \defeq \mathscr{Z}(\odiag,\hopf) \cdot \dim_{\field}(\hopf)^{g(\surface)-|\alphacurves| - |\betacurves|}.
	\label{eq:kuperberg}
\end{align}
Here $g(\surface)$ is the genus of the surface $\surface$, and $|\alphacurves|$ and $|\betacurves|$ denote the number of $\alpha$- and $\beta$-curves of $\odiag$, respectively. The ``complicated'' term $\mathscr{Z}(\odiag,\hopf) \in \field$ is defined as the \emph{evaluation} $\mathscr{Z}(\network)$ of a specific \emph{tensor network} $\network$ associated to $\odiag$, built from (several copies of) the \emph{structure tensors} of the Hopf algebra $\hopf$. We now very briefly and selectively review these notions.

\subparagraph*{Hopf algebras and their structure tensors} A \emph{Hopf algebra} $\hopf$ is a vector space $\hopfvec$ over a field $\field$ together with three $\field$-linear maps $\mult\colon\hopfvec \otimes \hopfvec \rightarrow \hopfvec$ (\emph{multiplication}), $\comult\colon\hopfvec\rightarrow \hopfvec \otimes \hopfvec$ (\emph{comultiplication}) and $\antipode\colon\hopfvec\rightarrow\hopfvec$ (\emph{antipode}) satisfying certain axioms (cf.\ \cite[Section~3]{kuperberg1991involutory}). We will assume that $\dim_{\field}(\hopfvec) < \infty$ and $\hopf$ is \emph{involutory}, i.e.,\ $\antipode \circ \antipode = \operatorname{id}_{\hopfvec}$. The maps $\mult$, $\comult$ and $\antipode$ can be regarded as tensors of type $(1,2)$, $(2,1)$ and $(1,1)$, respectively, see \Cref{fig:hopf-tensors}. 
\begin{figure}[h!]
	\centering
	\includegraphics[scale=1]{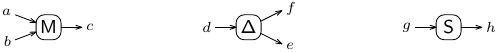}

	\caption{The multiplication, comultiplication and antipode tensors represented as \emph{coupons}.}
	\label{fig:hopf-tensors}
\end{figure}

\begin{example}
The \emph{group algebra} $\group = \field G$ of a finite group $G$ is naturally a Hopf algebra: the multiplication is defined as $\mult(g\otimes h) = gh$, the comultiplication as $\comult(g) = g \otimes g$, and the antipode as $\antipode(g) = g^{-1}$. In \cite[Sections 1 and 6]{kuperberg1991involutory} it is shown that $\sharp(\odiag,\group)$ equals the number of homomorphisms $\pi_1(\manifold) \rightarrow G$ from the fundamental group of the 3-manifold $\manifold$ described by $\odiag$ to the group $G$. (In this case $\sharp(\odiag,\group)$ is independent of the chosen orientation for $\diag$.)
\end{example}

\subsection{Kuperberg's construction and the proof of Theorem~\ref{thm:kuperberg-fpt}}
\label{ssec:kuperberg}

Following \cite[Section 5]{kuperberg1991involutory}, first we explain how the tensor network $\network$---whose evaluation $\mathscr{Z}(\network) \eqdef \mathscr{Z}(\odiag,\hopf)$ appears in \eqref{eq:kuperberg}---is constructed from an oriented Heegaard diagram $\odiag$. Then we prove \Cref{lem:kuperberg-width}, which---together with \Cref{thm:contraction}---implies \Cref{thm:kuperberg-fpt}.

\subparagraph*{The construction of $\bm{\network}$} As already mentioned, the building blocks of $\network$ are the structure tensors $\mult$, $\comult$ and $\antipode$ of the Hopf algebra $\hopf$ (\Cref{fig:hopf-tensors}). The shorthands in \Cref{fig:hopf-tensors-shorthand}---which are justified by the fact that both abbreviated tensor networks, called the \emph{tracial product} and \emph{tracial coproduct} in \cite{kuperberg1991involutory}, are cyclically symmetric\footnote{This is analogous to the fact that $\trace(A_1 A_2 \cdots A_n) = \trace(A_n A_1 \cdots A_{n-1})$ for the trace of matrices.}---will be useful:
\begin{figure}[ht]
	\centering
	\includegraphics[scale=1]{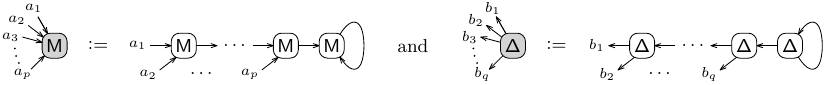}

	\caption{Diagrammatic abbreviations of the \emph{tracial product} and \emph{tracial coproduct} tensors.}
	\label{fig:hopf-tensors-shorthand}
\end{figure}

Let $\odiag = (\osurface,\oalpha,\obeta)$ be an oriented Heegaard diagram. For every oriented $\alpha$-curve $\vec\alpha_i \in \oalpha$, we take a copy of the tracial product tensor \raisebox{-1.65ex}{\includegraphics{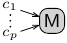}}, where the abstract indices $c_1,\ldots,c_p$ correspond to the crossings on $\vec\alpha_i$ (with the $\beta$-curves of $\odiag$) in the order they are encountered when traversing $\vec\alpha_i$ following its orientation. We assign a copy \raisebox{-2.15ex}{\includegraphics{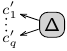}} of the tracial coproduct tensor to each oriented $\beta$-curve $\vec\beta_j \in \obeta$ analogously.
Note that each crossing $c$ between an $\alpha$- and a $\beta$-curve of $\odiag$ appears exactly twice as an abstract index of the above tensors: once at a product and once at a coproduct. The construction of $\network$ is concluded by contracting these indices (i.e.,\ joining the respective arrows with \emph{wires}) with a caveat: if the tangent vectors of the $\beta$- and $\alpha$-curves intersecting at $c$ form (in that order) a negatively oriented basis of the tangent space $T\osurface$, then the antipode tensor is interposed before contracting (\Cref{fig:contracting}).

\begin{figure}[h!]
	\centering
	\includegraphics[scale=1]{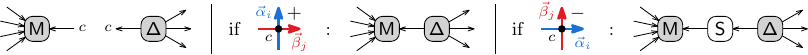}

	\caption{Instruction manual for assembling the tensor network $\network$ from its building blocks.}
	\label{fig:contracting}
\end{figure}

A fully contracted tensor network $\altnetwork$ (such as $\network$) can be seen as a graph $(V(\altnetwork),E(\altnetwork))$ in the obvious way: $V(\altnetwork)$ corresponds to the coupons of $\altnetwork$, whereas $E(\altnetwork)$ to its wires. (We emphasize that $\altnetwork$ is considered here without any diagrammatic abbreviations.)

\begin{lemma}
\label{lem:kuperberg-width}
Let $\odiag$ be an oriented Heegaard diagram and $\network$ be the Kuperberg tensor network induced by $\odiag$. Then, the number of vertices and the treewidth of $\odiag$ and $\network$ satisfy

\noindent\begin{subequations}
\parbox{0.32\textwidth}{%
\begin{align}
	|V(\network)| \leq 3|V(\odiag)|
	\label{eq:kuperberg-vertices}
\end{align}%
}\hfill%
\parbox{0.1\textwidth}{%
\begin{align*}
	~and
\end{align*}%
}\hfill%
\parbox{0.35\textwidth}{%
\begin{align}
	\tw(\network) \leq 2\tw(\odiag).
	\label{eq:kuperberg-width}
\end{align}%
}%
\end{subequations}
\end{lemma}

\begin{proof}
Both inequalities \eqref{eq:kuperberg-vertices} and \eqref{eq:kuperberg-width} follow directly from the construction of $\network$. In fact, the local configurations shown in \Cref{fig:tensor-network-local} already reveal the close structural relationship between $\odiag$ and $\network$. We refer to Appendix~\ref{app:kuperberg} for the complete proof. \qedhere
\begin{figure}[ht]
	\centering
	\begin{minipage}[t]{.30\textwidth}
		\centering
		\includegraphics[scale=1]{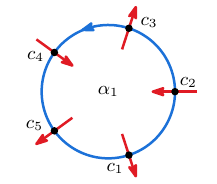}
		
		\subcaption{Local picture of a Heegaard diagram $\diag$ near an $\alpha$-curve $\alpha_1$}
		\label{fig:alpha-curve}
	\end{minipage}\hfill%
	\begin{minipage}[t]{.30\textwidth}
		\centering
		\includegraphics[scale=1]{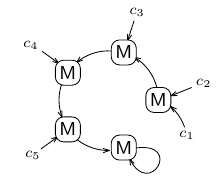}
		
		\subcaption{\begin{nolinenumbers}The corresponding local picture of the tensor network $\network$... \end{nolinenumbers}}
		\label{fig:network-local-alpha}
	\end{minipage}\hfill%
	\begin{minipage}[t]{.30\textwidth}
		\centering
		\includegraphics[scale=1]{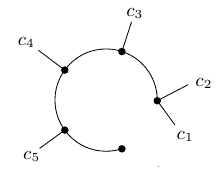}
		
		\subcaption{\begin{nolinenumbers}...and the local picture of its underlying simple graph $G_\network$ \end{nolinenumbers}}
		\label{fig:graph-local-alpha}
	\end{minipage}\hfill%

	\begin{minipage}[t]{.30\textwidth}
		\centering
		\includegraphics[scale=1]{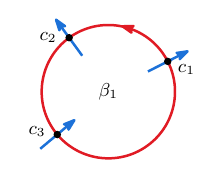}
		
		\subcaption{Local picture of a Heegaard diagram $\diag$ near an $\beta$-curve $\beta_1$}
		\label{fig:beta-curve}
	\end{minipage}\hfill%
	\begin{minipage}[t]{.30\textwidth}
		\centering
		\includegraphics[scale=1]{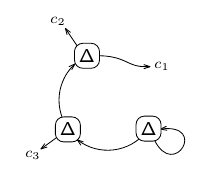}
		
		\subcaption{\begin{nolinenumbers}The corresponding local picture of the tensor network $\network$... \end{nolinenumbers}}
		\label{fig:graph-local-beta-simple}
	\end{minipage}\hfill%
	\begin{minipage}[t]{.30\textwidth}
		\centering
		\includegraphics[scale=1]{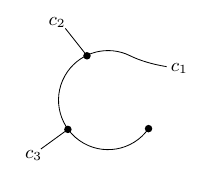}
		
		\subcaption{\begin{nolinenumbers}...and the local picture of its underlying simple graph $G_\network$ \end{nolinenumbers}}
		\label{fig:graph-local-beta}
	\end{minipage}\hfill%

	\caption{The local structural correspondence between a Heegaard diagram $\diag$ and the induced Kuperberg tensor network $\network$.}
	\label{fig:tensor-network-local}
\end{figure}
\end{proof}

\begin{proof}[Proof of \Cref{thm:kuperberg-fpt}] With \Cref{thm:heegaard-diag-tw} and \Cref{lem:kuperberg-width} at hand, \Cref{thm:kuperberg-fpt} is a consequence of the following general result about the computational cost of evaluating tensor networks.

\begin{theorem}[{\cite[Theorem 1]{ogorman2019parameterization}}, cf.\ {\cite[Theorem 4.6]{markov2008simulating}}]
\label{thm:contraction}
Any tensor network $\altnetwork$ with $n$ tensors can be evaluated in $O\big(2^{O(\vc(\altnetwork))}n\big)$ time, where $\vc(\altnetwork)$ denotes the \emph{vertex congestion} of $\altnetwork$.
\end{theorem}
Vertex congestion is a graph parameter similar to treewidth, discussed in \cite[p.\ 135]{bienstock1990embedding}. For a graph $G$ with maximum degree $\Delta$, we have $\vc(G) \leq (3/2)\cdot\Delta\cdot(\tw(G)+1)$, see \cite[Theorem~1 and Remark~3]{bienstock1990embedding}. Since Kuperberg's tensor networks are 3-regular graphs, we have

\begin{corollary}
\label{cor:contraction}
Any Kuperberg tensor network $\network$ can be evaluated in time $O\big(2^{O(\tw(\network))}|\network|\big)$, where $\tw(\network)$ denotes the treewidth and $|\network|$ the number of tensors of $\network$.
\end{corollary}

Now, if $\network$ is induced by an oriented triangulation $\otri$ with $n$ tetrahedra, then $|\network| \leq 3 |V(\odiag)|$ and $\tw(\network) \leq 2 \tw(\odiag)$  by \Cref{lem:kuperberg-width}, and $|V(\odiag)| \leq 6$ and $\tw(\odiag) \leq 12\tw(\otri)+11$ by \Cref{thm:heegaard-diag-tw}. Combining these inequalities with \Cref{cor:contraction}, the \Cref{thm:kuperberg-fpt} follows.
\end{proof}


\newpage

\bibliography{references}

\appendix

\section{On edge valence, experimentally}
\label{app:edge_valence_expe}

The treewidth of a triangulation is an indicator of {\em global sparsity}. While it measures how tree-like the gluing relation between tetrahedra is, it does not capture finer information about their proximity. Indeed, {\em tree decompositions} of triangulations relate only those tetrahedra glued along a common facet, and are blind to tetrahedra that are adjacent merely along an edge. The {\em valence} of an edge in a triangulation is the number of times it appears as an edge of some tetrahedron (counted with multiplicities). Although the average edge valence in a triangulated closed $3$-manifold is bounded, the valence may vary widely among edges; in particular, a single edge may be incident to all tetrahedra. In this sense, edge valence serves as an indicator of {\em local sparsity}.

\begin{figure}[ht]
\centering
    \begin{subfigure}[b]{0.55\textwidth}
        \centering
        \includegraphics[width=\textwidth]{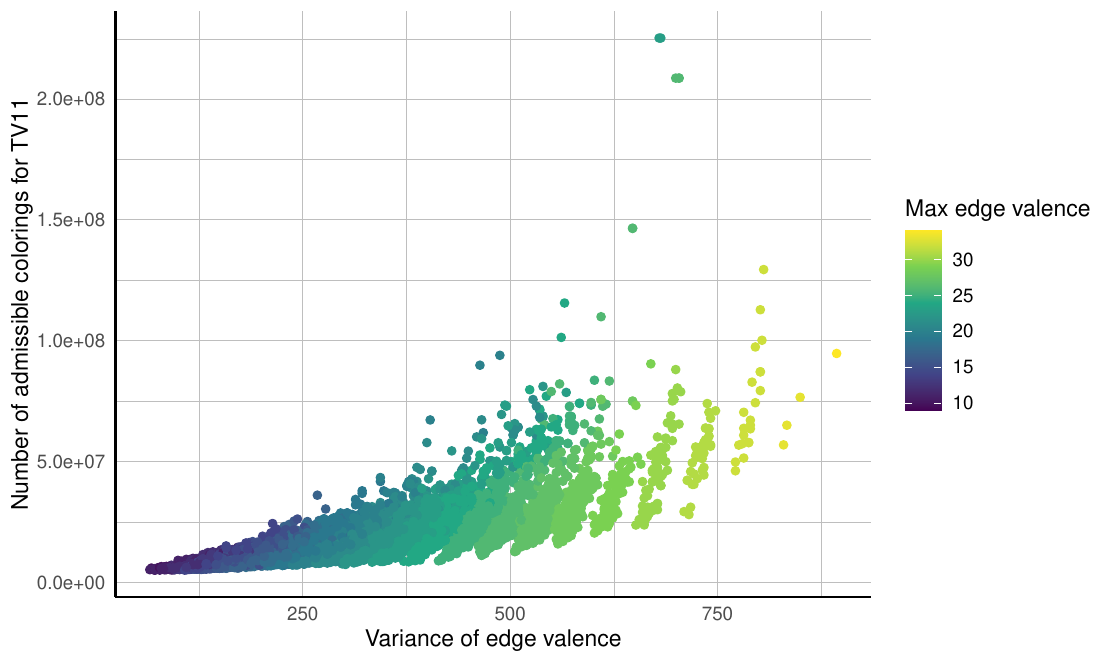}
    \end{subfigure}
    \hfill 
   \begin{subfigure}[b]{0.44\textwidth}
        \centering
        \includegraphics[width=\textwidth]{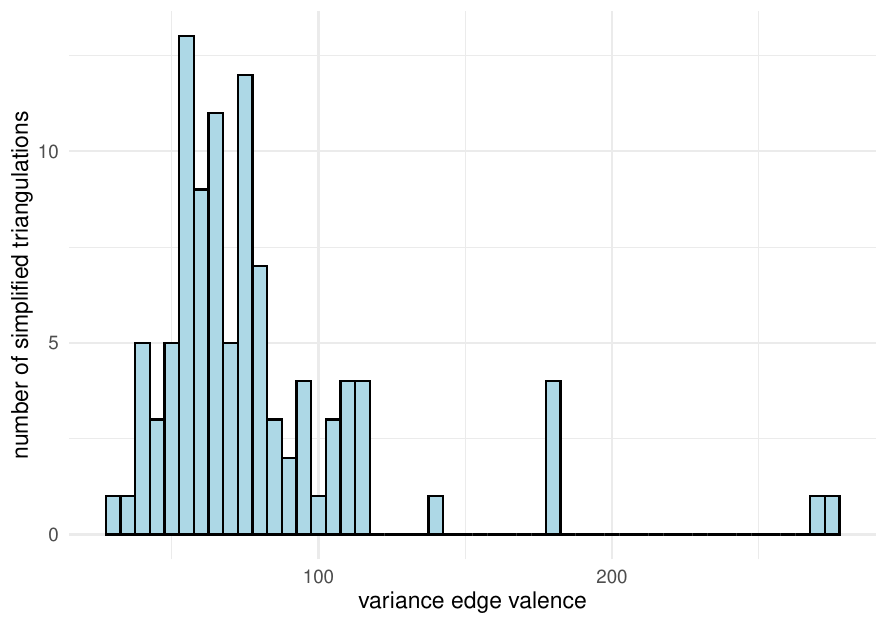}
    \end{subfigure}
\caption{Left: the number of admissible colorings for the Turaev--Viro quantum invariant $\operatorname{TV}_{11,1}(\operatorname{sl}_2(\mathbb{C}))$ for 10\,000, size 12, non-simplified triangulations of the lens space $L(11,1)$, as a function of the variance of edge valence. Right: 100 heuristically simplified triangulations of $L(11,1)$, size 8 to 13, plotted against the variance of their edge valence.}
    \label{fig:intro_plots1}
\end{figure}

\begin{figure}[ht]
\centering
    \begin{subfigure}[b]{0.55\textwidth}
        \centering
        \includegraphics[width=\textwidth]{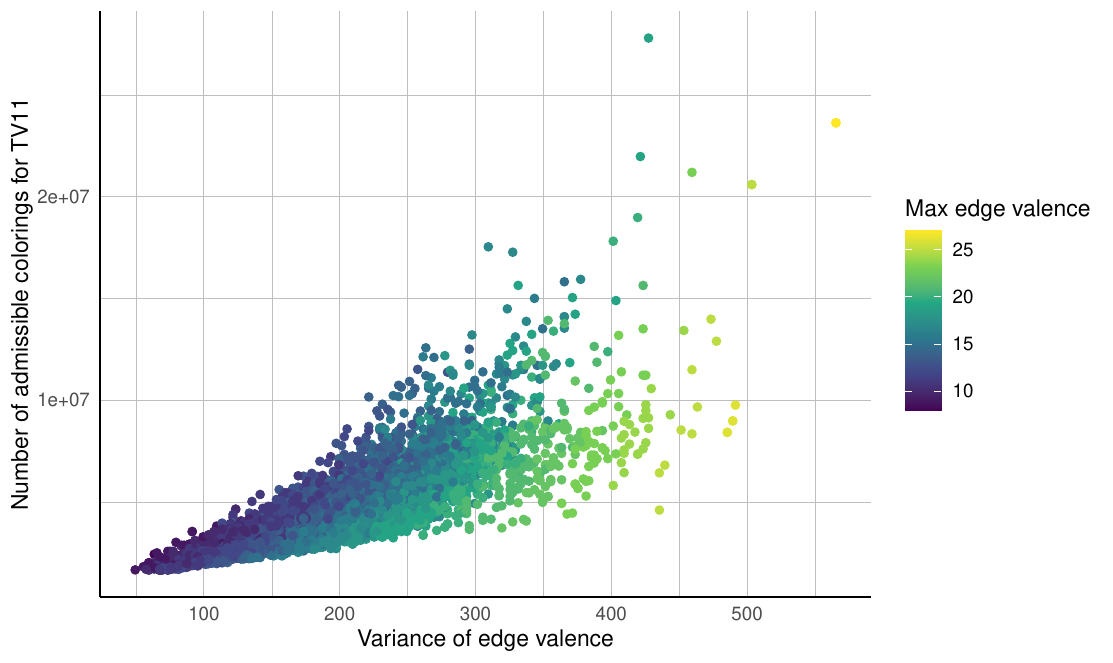}
    \end{subfigure}
    \hfill 
   \begin{subfigure}[b]{0.44\textwidth}
        \centering
        \includegraphics[width=\textwidth]{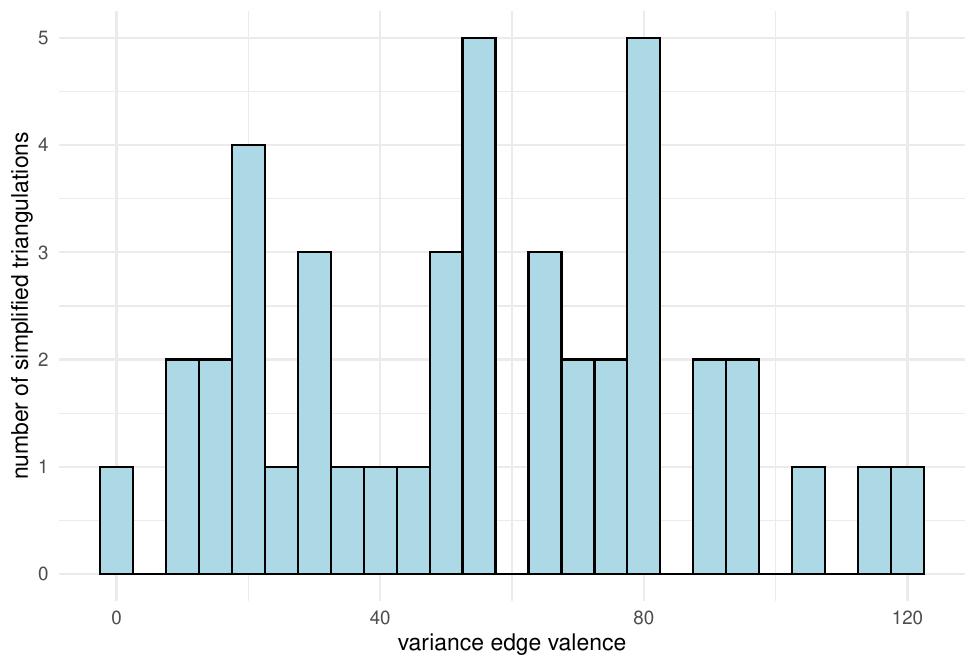}
    \end{subfigure}
\caption{Left: the number of admissible colorings for the Turaev--Viro quantum invariant $\operatorname{TV}_{11,1}(\operatorname{sl}_2(\mathbb{C}))$ for 10\,000, size 13, non-simplified triangulations of the Poincar\'e sphere, as a function of the variance of edge valence. Right: 43 heuristically simplified triangulations of the Poincar\'e sphere, size 5 to 11, plotted against the variance of their edge valence.}
    \label{fig:intro_plots2}
\end{figure}

In practice, high edge valences may significantly slow down computations. Many fundamental procedures---based on {\em normal surfaces} (e.g., sphere recognition), {\em hyperbolic geometry} (e.g., Thurston's gluing equations), or {\em colorings} (e.g., quantum invariants)---involve systems of equations that become increasingly interdependent as edge valences vary widely. To illustrate the practical consequences of this phenomenon, Figures~\ref{fig:intro_plots1} and~\ref{fig:intro_plots2} show the number of {\em admissible colorings} for the Turaev--Viro quantum invariant at $U_q(\operatorname{sl}_2(\mathbb{C}))$ as a function of the {\em variance} of edge valences in triangulations of the lens space $L(11,1)$ and the Poincar\'e sphere. Since the number of admissible colorings directly affects the running time of backtracking algorithms for computing Turaev--Viro invariants~\cite{maria2020computation}, this provides concrete evidence of the computational impact of edge-valence variability.

We also present the distribution of edge-valence variances among {\em reduced} triangulations of these two manifolds, which represent realistic algorithmic inputs. All triangulations were generated with {\tt Regina}~\cite{regina,burton2013regina}, and the reduced triangulations are local minima of {\tt Regina}'s simplification heuristics.


\section{Completion of the proof of Claim~\ref{claim:valid_tree_dec_retri}}
\label{app:missingproof}

Here we provide the omitted details of the proof of Claim~\ref{claim:valid_tree_dec_retri} from \Cref{sec:retri}. In particular we give the complete case analysis for verifying the ``edge coverage'' property of the tree decomposition $\treedecomp'$ of the triangulation $\tri^\ast$ defined in \Cref{lem:36tw}.

\begin{proof}[Proof of \emph{Edge coverage} in Claim~\ref{claim:valid_tree_dec_retri}] There are several types of tetrahedra that are adjacent in $\tri^\ast$. Let $a$ and $b$ be distinct vertices of $\tri(0)$ appearing as vertices of $\tri^\ast$. Let $\sigma_1$ and $\sigma_2$ be adjacent tetrahedra of $\tri^\ast$; we have, considering indices of the $u_i$ are taken modulo $k$, and indices of the $v_\ell$ are taken modulo $m$: all following configurations are illustrated in Figure~\ref{fig:retri_cases}.

  \begin{enumerate}[{\color{lipicsGray}\sffamily\bfseries (a)}]\itemsep0.5em
    \item If $\sigma_1$ and $\sigma_2$ are adjacent through a triangular face $\{x,y,z\}$ contained in a 2-face $F$, i.e.,\ $\sigma_1 = \{x,y,z,a\}$ and $\sigma_2 = \{x,y,z,b\}$, with $x,y,z \in \{u_i,v_j, w\}_{i,j}$, then $\sigma_1$ and $\sigma_2$ appear in a same bag $B'_\tau$, regardless of the type of triangle $\{x,y,z\}$.

    \item If $\sigma_1$ and $\sigma_2$ are adjacent through a triangle $\{a,v_\ell,u_i\}$ with $\sigma_1 = \{a,v_\ell,u_{i-1},u_i\}$ and $\sigma_2 = \{a,v_\ell,u_{i},u_{i+1}\}$ --- $u_{i-1}$ and $u_{i}$ are adjacent in $\dual(\tri)$ and must appear in a same bag $B_\tau$ of $\treedecomp$. By~\myref{itm:treedec1}, both $\sigma_1$ and $\sigma_2$ appear in $B'_\tau$ of $\treedecomp'$.
    \item If $\sigma_1$ and $\sigma_2$ are adjacent through a triangle $\{a,v_{\ell-1},u_{\ell(d-1)}\}$ with $\sigma_1 = \{a,v_{\ell-1},\allowbreak u_{\ell(d-1)-1},u_{\ell(d-1)})$ and $\sigma_2 = \{a,v_{\ell-1},u_{\ell(d-1)},v_\ell\}$ --- $u_{\ell(d-1)-1}$ and $u_{\ell(d-1)}$ are adjacent in $\dual(\tri)$ and must appear in a same bag $B_\tau$ of $\treedecomp$. By~\myref{itm:treedec1} for $\sigma_1$, and~\myref{itm:treedec2} for $\sigma_2$, both $\sigma_1$ and $\sigma_2$ appear in $B'_\tau$ of $\treedecomp'$. The case $\sigma_1 = \{a,v_{m-1},u_{k-1},u_{0}\}$ and $\sigma_2 = \{a,v_{m-1},u_{0},v_0\}$ works similarly.
    \item If $\sigma_1$ and $\sigma_2$ are adjacent through a triangle $\{a,v_{\ell},u_{\ell(d-1)})$ with $\sigma_1 = \{a,v_{\ell-1},\allowbreak u_{\ell(d-1)},v_\ell\}$ and $\sigma_2 = \{a,v_\ell,u_{\ell(d-1)},u_{\ell(d-1)+1}\}$ --- By~\myref{itm:treedec2} for $\sigma_1$ and~\myref{itm:treedec1} for $\sigma_2$, they both appear in any bag $B'_\tau$ such that $B_\tau$ contains $u_{\ell(d-1)}$.
    \item If $\sigma_1$ and $\sigma_2$ are adjacent through a triangle $\{a,v_{\ell-1},v_{\ell}\}$ with $\sigma_1 = \{a,w,v_{\ell-1},v_\ell\}$ and $\sigma_2 = \{a,v_{\ell-1},u_{\ell(d-1)},v_{\ell}\}$ --- By~\myref{itm:treedec3} for $\sigma_1$ and~\myref{itm:treedec2} for $\sigma_2$, they both appear in any bag $B'_\tau$ such that $B_\tau$ contains $u_{\ell(d-1)}$.
    \item If $\sigma_1$ and $\sigma_2$ are adjacent through a triangle $\{a,w,v_{\ell}\}$ with $\sigma_1 = \{a,w,v_{\ell-1},v_\ell\}$ and $\sigma_2 = \{a,w,v_{\ell},v_{\ell+1}\}$ --- $u_{(\ell+1)(d-1)-1}$ and $u_{(\ell+1)(d-1)}$ are adjacent in $\Gamma(\tri)$ and must appear in a same bag $B_\tau$ of $\treedecomp$. By~\myref{itm:treedec3}, both $\sigma_1$ and $\sigma_2$ appear in $B'_\tau$ of $\treedecomp'$.
    \item If $\sigma_1$ and $\sigma_2$ are adjacent through a triangle $\{a,u_i,u_{i+1}\}$ with $\sigma_1 = \{a,v_\ell,u_i,u_{i+1}\}$ and $\sigma_2 = \{a,v'_j,u_{i+1},u_i\}$, where $v_\ell$ and $v'_j$ are inner vertices of {\em distinct} 2-faces $F_1$ and $F_2$. By~\myref{itm:treedec1}, either $\sigma_1$ and $\sigma_2$ both appear in any bag $B'_\tau$ such that $B_\tau$ contains $u_{i}$ (if the orientations of the 2-faces are incompatible at the edge $\{u_i,u_{i+1}\}$) or $\sigma_1$ appears in the bags $B'_\tau$ such that $B_\tau$ contains $u_{i}$ and $\sigma_2$ appears in the bags $B'_\tau$ such that $B_\tau$ contains $u_{i+1}$. However, for the latter, $u_i$ and $u_{i+1}$ are adjacent in $\dual(\tri)$ and there must be a bag $B_\tau$ containing both $u_i$ and $u_{i+1}$ in $\treedecomp$; in consequence $\sigma_1,\sigma_2 \in B'_\tau$ in $\treedecomp$. \qedhere
  \end{enumerate}
\end{proof}

\newpage

\section{Some useful lemmas for Section~\ref{sec:retri}}
\label{app:retri}


\subsection{Iterated radicals, limit and speed of convergence}
\label{ssec:iterated}

\begin{lemma}[Iterated radicals]
\label{lem:comp1}
Let $c$ be a positive constant, and let $L = c + \frac{1}{2} + \sqrt{c+\frac{1}{4}} > 0$. Consider the sequence $a_0 > L$, $a_{n+1} = \sqrt{a_n} + c$. Then $(a_n)_{n\geq 0}$ is a decreasing sequence converging to $L$. Additionally, for every $n \in \mathbb{N}$ we have the bound:
\[
	0 \leq a_n \leq \displaystyle a_0^{1/2^n} + L.
\]
In particular, for $c=4$, this gives $L < 6.32$.
\end{lemma}

\begin{proof}
The quantity $L$ is computed as a solution to the equation $x = \sqrt{x}+c$, obtained by studying the degree two polynomial $L=(L-c)^2$. We prove the bound by induction. The property holds for $a_0$. Suppose $a_n \leq a_0^{\frac{1}{2^n}} + L$, we have:
\begin{align*}
  \displaystyle a_{n+1} = \sqrt{a_n} + c &\leq \sqrt{ \left(a_0^{1/2^n} + L\right)} + c \hfill&\text{(by induction)}\\
                 &\leq a_0^{1/2^{n+1}} + \sqrt{L} + c \hfill&\text{(by concavity of}\ \sqrt{(\cdot)})\\
                 &= a_0^{1/2^{n+1}} + L. &(\text{by using}\ L=\sqrt{L}+c) \tag*{\qedhere}
\end{align*}
\end{proof}

Next, we evaluate how fast $a_n$ converges within a small neighborhood of its limit.

\begin{lemma}
For $n \geq \log_2 \log_2 x$, $|a_n - \lim a_n| \leq 2$. 
\label{lem:limit} 
\end{lemma}

\begin{proof}
We compute $n$ such that $a_n \leq a_0^{1/2^n} + L \leq L+2$. Under the hypothesis, 
\[
  a_0^{1/2^n} \leq 2 \Longleftrightarrow 2^{\frac{\log_2 a_0}{2^n}} \leq 2 
                           \Longleftrightarrow \frac{\log_2 a_0}{2^n} \leq 1
                           \Longleftrightarrow \log_2 a_0 \leq 2^n
                           \Longleftrightarrow \log_2 \log_2 a_0 \leq n. \hfill \qedhere
\]
\end{proof}

\begin{lemma}
Consider the sequences $x_n,y_n$ such that $x_0,y_0 > 0$ and $x_{n+1} = (28+4\sqrt{6})x_{n}+(16+4\sqrt{6})y_n$ and $y_{n+1} = (6+\sqrt{6})(x_{n}+y_n)$. Then, setting $\lambda = 30+4\sqrt{6}$ gives,
\[
    x_n < c_1 \lambda^n (x_0+y_0) + o(\lambda^n) \quad \text{and} \quad y_n < c_2 \lambda^n(x_0+y_0) + o(\lambda^n),
\]
for $c_1,c_2 > 0$ constants.
\label{lem:doublesequence}
\end{lemma}

\begin{proof}
Consider
\[
    \left(\begin{array}{c} x_n\\ y_n\\ \end{array}\right) = A^n \left(\begin{array}{c} x_0\\ y_0 \\ \end{array}\right),~~\text{with}~~A = \left(\begin{array}{cc} 28+4\sqrt{6}& 16+4\sqrt{6}\\ 6+\sqrt{6} & 6+\sqrt{6} \\\end{array}\right).
\] 
Diagonalizing $A$, that is, writing $A = PDP^{-1}$ with:
\[
P = \left(\begin{array}{cc} 8+\sqrt{6}& 2\\ 1 & -(6+\sqrt{6}) \\\end{array}\right), P^{-1} = \frac{1}{62+20\sqrt{6}}\left(\begin{array}{cc} 6+\sqrt{6}& 2\\ 1 & -(8+2\sqrt{6}) \\\end{array}\right), 
\]
\[
\text{and} \  D = \left(\begin{array}{cc} 30+4\sqrt{6}& 0\\ 0 & 4+\sqrt{6} \\\end{array}\right),
\]
we deduce the desired result by applying $A^n = P^{-1}D^nP$ to $(x_0,y_0)^\top$.
\end{proof}

\subsection{Counting edge valences in triangulated 3-manifolds}
\label{ssec:countval}

\begin{lemma}[Euler characteristic argument]
\label{lem:eulerchi}
In any triangulation $\tri$ of a closed $3$-manifold $\M$ with $v$ vertices, $e$ edges, $f$ faces, and $n$ tetrahedra, we have $e = v+n$. 
\end{lemma}

\begin{proof}
As $\M$ is closed, every face of $\tri$ is shared by exactly two tetrahedra, thus $f = 2n$. Every odd-dimensional closed manifold has Euler characteristic zero, see, e.g. \cite[Corollary~3.37]{hatcher2002algebraic}, hence $v-e+f-n = 0$. We conclude $e=n+v$.
\end{proof}

\begin{corollary}[Sum of edge valences]
\label{lem:sumval}
In any triangulation $\tri$ of a closed $3$-manifold $\M$ with $n$ tetrahedra, the sum $\mathfrak{S} = \sum_{e\in\tri(1)}\val(e)$ of edge valences equals $6n$.
\end{corollary}

\begin{proof}
The sum $\mathfrak{S}$ counts, with multiplicity, the incidences between edges and faces. Since every face is triangular, it follows that $\mathfrak{S} = 3\cdot|\tri(2)|$. By \Cref{lem:eulerchi} we obtain $\mathfrak{S} = 6n.$
\end{proof}

\begin{lemma}[Sum of square roots of edge valences]
\label{lem:sumvalsqrt}
In any triangulation $\tri$ of a closed $3$-manifold $\M$ with $\vertices$ vertices and $n$ tetrahedra, we have $\sum_{e\in \tri(1)}\sqrt{\val(e)} \leq \sqrt{6}(n+\vertices)$.
\end{lemma}

\begin{proof}
By the Cauchy--Schwarz inequality, any positive real numbers $a_1, \ldots, a_n$ satisfy $\sqrt{a_1}+\ldots \sqrt{a_n} \leq \sqrt{n}\sqrt{a_1+\ldots a_n}$. Applying this to the edge valences we get
\[
\sum_{e\in \tri(1)}\sqrt{\val(e)} \underbrace{\leq}_{\text{Cauchy--Schwarz}} \sqrt{|\tri(1)|}\sqrt{\sum_{e\in \tri(1)} \val(e)} \underbrace{\leq}_{\text{Lem.\ \ref{lem:eulerchi}, Cor.\ \ref{lem:sumval}}} \sqrt{n+v}\sqrt{6n} \leq \sqrt{6}(n+\vertices).  
\]
\end{proof}

\section{Very briefly on tensor networks}
\label{app:tensor}


Originating in physics \cite[Section~1.3]{hess2015tensors}, tensors generalize vectors and linear transformations to a multilinear setting. Tensor networks play an important role in quantum computing \cite{berezutskii2025tensor}. 
We refer to \cite[Section~2]{kuperberg1991involutory} for a tailored introduction and to \cite{biamonte2017nutshell} for more details.

\subparagraph*{Tensors} Let $\vectorspace$ be a finite-dimensional vector space over a field $\field$. A \emph{tensor} $\tensor$ is an element of some tensor product space $\vectorspace_1 \otimes \dots \otimes \vectorspace_{r}$, where $\vectorspace_i = \vectorspace$ or $\vectorspace_i = \dualspace$ for each $i$. If $p$ is the number of $\vectorspace$-factors in this product, then we say $\tensor$ is of \emph{type} $(p,r-p)$. We let $T^p_q(\vectorspace)$ denote the set of tensors of type $(p,q)$.

\begin{example}
Tensors of type $(0,0)$, $(1,0)$, $(0,1)$, $(1,1)$, and $(0,2)$ correspond to scalars, vectors, linear functionals, linear transformations, and bilinear forms, respectively. An example for a $(1,2)$-tensor is the cross product $(v,w) \mapsto v \times w$ of two vectors in $\R^3$.
\end{example}

Two tensor operations are essential here. First, given two tensors, $\tensor$ of type $(p,q)$ and $\alttensor$ of type  $(r,s)$, their \emph{tensor product} $\tensor \otimes \alttensor$ is a tensor of type $(p+r,q+s)$. Second, if $p,q > 0$, then for any $i \in \{1,\ldots,p\}$ and $j \in \{1,\ldots,q\}$ the \emph{$(i,j)$-contraction} map $C_{i,j}\colon T^p_q(\vectorspace) \rightarrow T^{p-1}_{q-1}(\vectorspace)$ is defined by applying the canonical map $\vectorspace \otimes \dualspace \rightarrow \field$, $v \otimes f \mapsto f(v)$ to the $i\textsuperscript{th}$ $\vectorspace$-factor and the $j\textsuperscript{th}$ $\dualspace$-factor. E.g.,\ if $\linmap \in T^{1}_{1}(\vectorspace)$ then $C_{1,1}(\linmap) = \operatorname{Tr}(\linmap)$ is just the \emph{trace} of $\linmap$.

\subparagraph*{Index notation} The \emph{(abstract) index notation} provides a succinct way to describe basis-independent computations with tensors. In this notation, for example, a tensor $\tensor \in \vectorspace \otimes \vectorspace \otimes \dualspace \otimes \vectorspace \otimes \dualspace$ is written as $\tensor^{{a_1}{a_2}}{}_{a_3}{}^{a_4}{}_{a_5}$, were $a_i$ is an upper (resp.\ lower) index if $\vectorspace_i = \vectorspace$ (resp.\ $\vectorspace_i = \dualspace$). These abstract indices $a_i$ are pairwise different symbols that merely indicate the slots of $\tensor$---they are \textbf{not} numerical indices. 
The tensor product of two tensors, say $\tensor^{a}{}_{bc}$ and $\alttensor^{d}{}_{e}$, is written by juxtaposition: $\tensor^{a}{}_{bc}\alttensor^{d}{}_{e}$. Performing an $(i,j)$-contraction is indicated by setting the $i\textsuperscript{th}$ upper index and the $j\textsuperscript{th}$ lower index to be equal, e.g.,\  $C_{1,2}(\tensor^{a}{}_{bc}\alttensor^{d}{}_{e}) = \tensor^{a}{}_{ba}\alttensor^{d}{}_{e}$, or $C_{2,1}(\tensor^{a}{}_{bc}\alttensor^{d}{}_{e}) = \tensor^{a}{}_{bc}\alttensor^{b}{}_{e}$. Tensor products and contractions encompass many standard operations in linear algebra, e.g.,\ $\linmap^a{}_b \vect^b$ equals the $\linmap$-image of the vector $\vect$ (cf.\  \Cref{fig:tensor-networks-linalg}).\footnote{Moreover, using the \emph{Einstein summation convention}, i.e.,\ summing over the indices appearing twice, $\linmap^a{}_b \vect^b = \linmap (\vect)^a$ exactly describes how to calculate the coordinates of $\linmap(\vect)$, once a basis for $\vectorspace$ is fixed.}

\subparagraph*{Tensor diagrams} Importantly, the index notation can be developed into a \emph{diagrammatic notation}. To construct the \emph{diagram} of a tensor $\tensor$ given by index notation, first take a \emph{coupon} \raisebox{-.8ex}{\includegraphics[scale=.95]{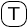}}. Then, going around \raisebox{-.8ex}{\includegraphics[scale=.95]{coupon}} in counter-clockwise direction, for each index $a_i$ attach to \raisebox{-.8ex}{\includegraphics[scale=.95]{coupon}} an outgoing (resp.\ incoming) $a_i$-labeled arrow, whenever $a_i$ is an upper (resp.\ lower) index. The diagram of $\tensor \otimes \alttensor$ is the disjoint union of the diagrams of $\tensor$ and $\alttensor$, and a contraction of an upper-lower index pair is indicated by joining the respective arrows with a \emph{wire} and erasing the common label from the drawing, cf.\ \Cref{fig:tensor-networks-linalg,fig:tensor-operations}. Repeatedly performing these operations gives rise to what we call a \emph{tensor network}. If a tensor network $\altnetwork$ has no labeled arrows---because each arrow is incident to two coupons (which may coincide)---then $\altnetwork$ is \emph{fully contracted}. In this case $\altnetwork$ represents a scalar $\mathscr{Z}(\altnetwork) \in \field$, which we call the \emph{evaluation}\footnote{The word \emph{contraction} is also commonly used to denote the evaluation of a tensor network.} of $\altnetwork$. 

\begin{figure}[ht]
	\centering
	\includegraphics[scale=1]{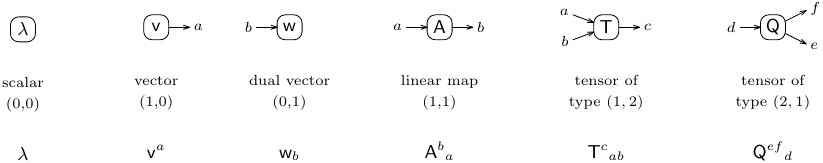}

	\caption{Representing tensors using the index notation (bottom) and via diagrams (top).}
	\label{fig:tensor-notation}
\end{figure}

\begin{figure}[ht]
	\centering
	\includegraphics[scale=1]{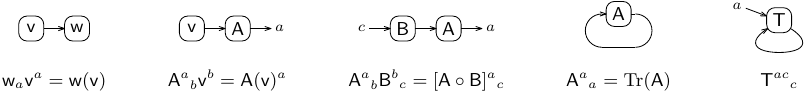}

	\caption{Many basic operations in linear algebra can be represented by small tensor networks.}
	\label{fig:tensor-networks-linalg}
\end{figure}

\begin{figure}[ht]
	\centering
	\includegraphics[scale=1]{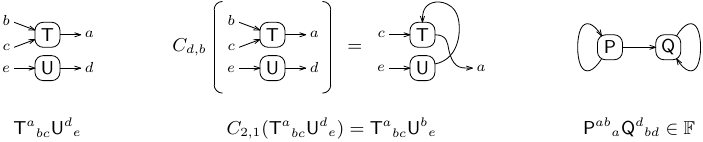}

	\caption{Tensor operations in index notation (bottom) and via tensor networks (top).}
	\label{fig:tensor-operations}
\end{figure}

\newpage

\section{The complete proof of Lemma~\ref{lem:kuperberg-width}}
\label{app:kuperberg}


\begin{proof}[Proof of {\Cref{lem:kuperberg-width}}]
As stated in \Cref{ssec:kuperberg}, both inequalities \eqref{eq:kuperberg-vertices} and \eqref{eq:kuperberg-width} follow directly from the construction of $\network$.

To show \eqref{eq:kuperberg-vertices}, note that $|V(\network)| = \#\raisebox{-.8ex}{\includegraphics[scale=.95]{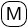}} + \#\raisebox{-.8ex}{\includegraphics[scale=.95]{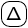}} + \#\raisebox{-.8ex}{\includegraphics[scale=.95]{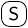}}$, where $\#\raisebox{-.8ex}{\includegraphics[scale=.95]{mcoupon}}$, $\#\raisebox{-.8ex}{\includegraphics[scale=.95]{dcoupon}}$ and $\#\raisebox{-.8ex}{\includegraphics[scale=.95]{scoupon}}$ count the multiplication, comultiplication and antipode tensors comprising $\network$. Inspecting \Cref{fig:hopf-tensors-shorthand}, we see that the number of labeled arrows of the tracial (co)product tensor equals  the number of its constituent (co)multiplication tensors. Now, in the construction of $\network$ the labels correspond to the vertices (i.e.,\ crossings) of $\odiag$. As each label appears once at an \raisebox{-.8ex}{\includegraphics[scale=.95]{mcoupon}}-coupon and once at a \raisebox{-.8ex}{\includegraphics[scale=.95]{dcoupon}}-coupon, we obtain $\#\raisebox{-.8ex}{\includegraphics[scale=.95]{mcoupon}} + \#\raisebox{-.8ex}{\includegraphics[scale=.95]{dcoupon}} = 2|V(\odiag)|$. The inequality $\#\raisebox{-.8ex}{\includegraphics[scale=.95]{scoupon}} \leq |V(\odiag)|$ is immediate, since $V(\odiag)$ corresponds to all crossings, while the \raisebox{-.8ex}{\includegraphics[scale=.95]{scoupon}}-coupons to the negative crossings of $\odiag$.

To show \eqref{eq:kuperberg-width}, first note that treewidth is invariant under subdividing edges (see, e.g., \cite[Section~3.6]{gurski2025behavior}) thus the presence of antipode tensors does not affect the treewidth of $\network$. Hence, we may assume that $\network$ does not contain any antipode tensors. As for \raisebox{-.8ex}{\includegraphics[scale=.95]{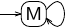}} and \raisebox{-.8ex}{\includegraphics[scale=.95]{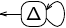}}, the presence of these tensors in $\network$ correspond to degree-one vertices in its \emph{underlying simple graph} $G_\network$, see \Cref{fig:graph-local-alpha,fig:graph-local-beta}. Since treewidth is not sensitive to loops or multiple edges either, we get that $\tw(G_\network) = \tw(\network)$. Let $G_\network^\ast$ be the subgraph of $G_\network$ obtained by erasing all such degree-one vertices in $G_\network$. Note that $\tw(G_\network^\ast) = \tw(G_\network)$, since attaching degree-one vertices to a graph does not change its treewidth. Hence we are left with bounding $\tw(G_\network^\ast)$.

Let $\treedecomp_{\odiag} = (T,\bags)$ be an optimal tree decomposition of $\odiag$, and let $V(\odiag) = \{c_1,\ldots,c_m\}$. By construction of $\network$, each vertex of $G_\network^\ast$ is incident to precisely one $c_i$-labeled edge and those labels are unique. We define the desired tree decomposition $\treedecomp' = (T',\bags')$ of $G_\network^\ast$ as follows:
\begin{itemize}
	\item $T' \defeq T = (I,F)$, i.e., the tree structure of the decomposition is the same as for $\treedecomp_{\odiag}$.
	\item $\bags' \defeq \{B'_i : i \in I\}$ and for each $i \in I$ we obtain $B'_i$ by first taking $B_i \in \bags$, then replacing each vertex $c \in B_i$ with the two vertices of $G_\network^\ast$ that are connected by the $c$-labeled edge.
\end{itemize}
Clearly, $|B'_i| = 2|B_i|$ for each $i \in I$. To see that $\treedecomp' = (T',\bags')$ is a valid tree decomposition of $G_\network^\ast$, note the following:
\begin{enumerate*}
	\item Every vertex $v \in V(G_\network^\ast)$ is contained by some bag of $\bags'$, since $v$ is incident to at least one labeled edge.
	\item For every edge $e = \{u,v\} \in E(G_\network^\ast)$, if $e$ is labeled, then $u$ and $v$ are readily contained in some common bag of $\bags'$. If $e$ is not labeled, then let $c_u$ and $c_v$ be the labels of the labeled edges incident to $u$ and $v$, respectively. Note that $c_u$ and $c_v$ are successive crossings along some $\alpha$- or $\beta$-curve of $\odiag$. Hence by the edge coverage property of $\treedecomp_{\odiag}$ there is a bag $B_i \in \bags$ that contains both $c_u$ and $c_v$. But then $B'_i$ contains, by definition, both $u$ and $v$.
	\item By construction of $\treedecomp'$, the subtree property is automatic.
\end{enumerate*}
\end{proof}

\end{document}